\documentclass{article}

\usepackage[a4paper, margin=1in]{geometry} 
\usepackage{setspace}
\DeclareUnicodeCharacter{00A0}{ }
\usepackage{graphicx}
\usepackage[english]{babel}
\usepackage{hyperref}
\usepackage{subcaption}
\usepackage{comment}
\usepackage{caption}
\usepackage{float}
\usepackage{flafter}  
\usepackage{placeins}

\usepackage{amsmath}
\usepackage{mathtools,amssymb,bm,mathrsfs}
\usepackage{amsfonts}
\usepackage{amsthm}
\usepackage[nopatch]{microtype}
\usepackage{booktabs}

\usepackage[numbers]{natbib}
\newcommand{\autocite}[1]{\cite{#1}}

\usepackage{authblk}

\usepackage{tikz}
\usetikzlibrary{backgrounds}
\usetikzlibrary{external}
\usetikzlibrary{decorations.shapes,arrows,calc,decorations.markings,shapes,decorations.pathreplacing,shapes.arrows}
\usetikzlibrary{positioning}
\tikzset{main node/.style={circle,fill=blue!20,draw,inner sep=1pt},}

\newtheorem{theorem}{Theorem}
\newtheorem{definition}{Definition}
\newtheorem{example}{Example}
\newtheorem{lemma}{Lemma}

\theoremstyle{remark}
\newtheorem{remark}{Remark}

\newcommand{\G}{\mathcal{G}}
\newcommand{\E}{\mathcal{E}}
\newcommand{\V}{\mathcal{V}}
\newcommand{\limt}{\lim_{t\rightarrow\infty}}
\newcommand{\piz}{\pi^{*}}
\newcommand{\pie}{\pi^{(E)}}
\newcommand{\map}[3]{#1:#2 \rightarrow #3}

\newcommand{\real}{{\mathbb{R}}}

\newcommand{\Tresmaprand}{\Phi}
\newcommand{\bg}{\mathbf{G}}
\newcommand{\bv}{\mathbf{V}}
\newcommand{\be}{\mathbf{E}}

\newcommand{\bb}{\mathbf{B}}
\DeclareMathOperator{\Var}{Var}

\allowdisplaybreaks

\title{P\'olya Threshold Graphs\thanks{
This work was funded in part by the Natural Sciences and Engineering Research Council (NSERC) of Canada.
The authors are with the Department of Mathematics and Statistics, Queen's University, Kingston, ON, Canada (Emails: \{19jy98,fa,bahman.gharesifard\}@queensu.ca).}
}

\author{Jinghan Yu, Fady Alajaji, Bahman Gharesifard}

\date{}

\begin{document}

\maketitle

\begin{abstract}
We introduce the P\'olya threshold graph model and derive its stochastic and algebraic properties. This random threshold graph is generated sequentially via a two-color P\'olya urn process. Starting from an empty graph, each time step involves a draw from the urn that produces an indicator variable, determining whether a newly added node is universal (connected to all existing nodes and itself) or isolated (connected to no existing nodes). This construction yields a random threshold graph with an adjacency matrix that admits an explicit representation in terms of the draw sequence. Using the structure of the P\'olya draw process, we derive the exact degree distribution for any arbitrary node, including its mean and variance. Furthermore, we evaluate a distance-based decay centrality score and provide an explicit expression for its expectation. On the algebraic side, we explicitly characterize the Laplacian matrix of the random threshold graph, obtaining a closed-form description of its spectrum and corresponding eigenbasis. Finally, as an application of these structural results, we analyze discrete-time consensus dynamics on P\'olya threshold graphs.

\vspace{1em}
\end{abstract}

\section{Introduction}
In this paper, we define a class of threshold graphs in which nodes and edges are generated randomly and sequentially via a two-color P\'olya urn process. At each time step, a draw from the urn determines whether the newly added node is an isolated node or a universal node, and each realization of this process produces a threshold graph. Such a construction yields a random threshold graph that provides a reinforcement-driven dependent extension of sequential threshold graph generation.

Threshold graphs have been an important topic in both graph theory and  network science. Following~\autocite{golumbic2004algorithmic, NM-UP:95, bapat2013adjacency}, threshold graphs admit several equivalent characterizations, including a weight and threshold inequality and a sequential binary construction formed by repeatedly adding either universal or isolated nodes. They serve as fundamental models for systems exhibiting a nested structure and specific dynamic properties, see~\autocite{AH-PS-DS:06}. Furthermore, the rich algebraic structure of deterministic threshold graphs has motivated extensive studies on their Laplacian spectra, including the characterization of Laplacian controllability for multi-agent networks~\autocite{Aguilar2015}, spanning trees~\autocite{bapat_spanning_trees}, critical groups~\autocite{critical_group_threshold}, and maximum Laplacian energy~\autocite{max_laplacian_energy}.

It is natural to consider random versions of threshold graphs. \citet{Reilly2000} have introduced natural equivalent models for random threshold graphs and used them to derive a variety of graph-theoretic properties. Diaconis \textit{et al.} further studied the threshold graph limits and placed several random threshold graph models into a broader probabilistic framework~\cite{diaconis2008threshold}. Subsequent work has also examined structural questions such as connectivity in random threshold graph models~\autocite{6544546}.

Motivated by prior works on random threshold graphs, we introduce the class of P\'olya threshold graphs, in which the binary node type construction is generated by a reinforced P\'olya urn process rather than by an independent sampling rule. This is a natural choice because the sequential construction of threshold graphs is binary; meanwhile, the P\'olya urn process also provides a canonical way to generate dependent sequences through reinforcement. Recent work has also shown that the P\'olya urn process can be used to generate other classes of random graphs, such as preferential attachment graphs~\autocite{Singh_2024}. As a result, the model remains analytically tractable while allowing the node types to be generated through the dependent random process.

Exploiting this tractability, our work is structured into three main parts, after providing background material on the P\'olya urn model and random threshold graphs in Section~\ref{sec:background}. First, in Section~\ref{section:properties}, we present an explicit adjacency matrix representation and characterize the stochastic properties of the P\'olya threshold graph, deriving the exact degree distribution for an arbitrary node $V_i$, together with its mean and variance, as well as an expression for the expected decay centrality score. Second, also in Section~\ref{section:properties}, we establish the algebraic properties of the model, including the Laplacian matrix, the Laplacian spectrum and a corresponding eigenbasis.  Third, in Section~\ref{section:consensus}, we study the discrete-time consensus dynamics on connected P\'olya threshold graphs, compare the theory with simulations, and explore an extension to the finite-memory P\'olya urn setting. Finally, we conclude the paper in Section~\ref{sec:conclusion}. 

\section{Preliminaries}\label{sec:background}

\subsection{P\'olya urn model}\label{polya-urn}

We start by recalling the setup for the classical P\'olya urn model~\autocite{polya1930}. Consider an urn equipped with red and black balls. At the initial time $t = 0$, the urn has $R>0$ red balls and $B>0$ black balls, for a total of $T=R+B$ initial balls. At each discrete time step $ t \geq 1 $, we draw a ball from the urn at random, and place it back along with $\Delta >0$ extra balls of the same color. We allow the number of balls to be a positive real number, since draws depend only on color proportions; allowing positive real ball counts does not alter the model. We naturally define the draw indicator random variable $Z_t$ at time $t $ as follows
\begin{align}
    Z_t = 
    \begin{cases}
        1& \text{if draw at time $t$ is red,}\\
        0 & \text{if draw at time $t$ is black.} 
    \end{cases}
    \label{equation:draw}
\end{align}
It is well-known that the stochastic process $\{Z_t\}_{t=1}^\infty$ is \emph{exchangeable} (and thus strictly stationary), see~\autocite{polya1930,alajaji2002communication}, where by exchangeability we mean that for any positive integer $\ell$, the draw variables $Z_1,Z_2,\ldots,Z_\ell$ have an invariant joint distribution with respect to all permutations of the indices $1,2,\ldots,\ell$. 
Specifically, following~\autocite{alajaji2002communication}, the joint probability of observing a specific length-$n$ sequence $z^n =(z_1,\ldots, z_n) \in \{0,1\}^n$ with exactly $k = \sum_{i=1}^n z_i$ red draws is given by
\begin{align}\label{jointprobofPolya_product}
    P(Z_1 = z_1,\cdots,Z_n = z_n) = \frac{\prod_{i=0}^{k-1}(R+i\Delta)\prod_{j=0}^{n-k-1}(B+j\Delta)}{\prod_{m=0}^{n-1}(R+B+m\Delta)}.
\end{align}
By defining the urn's initial red balls proportion $\rho = \frac{R}{R+B}$ and the reinforcement parameter $\delta = \frac{\Delta}{R+B}$, the joint distribution above can be written as 
\begin{align}\label{jointProbofPolya}
    P(Z_1 = z_1,\cdots,Z_n = z_n) = \frac{\Gamma(\frac{1}{\delta})\Gamma(\frac{\rho}{\delta}+k)\Gamma(\frac{1-\rho}{\delta}+n-k)}{\Gamma(\frac{\rho}\delta)\Gamma(\frac{1-\rho}{\delta})\Gamma(\frac{1}{\delta}+n)},
\end{align}
where $\Gamma(x)=\int_0^\infty t^{x-1}e^{-t}dt$ is the \textit{Gamma function}.
The above distribution exhibits a Beta-Binomial structure, which will be instrumental in deriving the properties of our model in Section~\ref{section:properties}.

The P\'olya urn-based draw process $\{Z_t\}_{t=1}^\infty$ enjoys many other interesting properties, including the fact that 
its sample average $\frac{1}{n} \sum_{i=1}^nZ_i$ for $n\in\mathbb{Z}_{>0}$ converges almost surely as $n$ grows without bound to a random variable governed by the Beta distribution, $\text{Beta}\big(\frac{\rho}{\delta}, \frac{1-\rho}{\delta}\big)$, with parameters $\frac{\rho}{\delta}$ and $ \frac{1-\rho}{\delta}$; in particular, this implies that the draw process is non-ergodic, see~\autocite{polya1930}.
 
\subsection{Random threshold graphs}

In this section, we introduce \textit{threshold graphs}, following~\autocite{NM-UP:95}. An undirected graph $\bg_n$ with $n$ nodes is a pair $(\bv_n,\be_n)$, where $\bv_n$ is the node set of size $n$, and $\be_n \subseteq \bv_n \times \bv_n$ is the edge set. Moreover, we assume that a graph can have self-loops, i.e., an edge from a node $v$ to itself so that $(v,v)\in\be_n$. Throughout this paper, all graphs are undirected; when we write $(u,v)\in \be$, we also mean $(v,u)\in \be$.

\begin{definition}
    Given a size-$n$ graph $\bg_n = (\bv_n,\be_n)$, where $\bv_n =\{v_1,\ldots,v_n\}$, the  adjacency matrix $A_n$ of $\bg_n$ is an $n\times n$ matrix, defined as 
    \begin{align}
    A_n = [a_{ij}],
    \end{align}
    where for $i,j= 1,2,\ldots, n$,
    \begin{align}
        a_{ij}:=
        \begin{cases}
            1  &\text{if $(v_i,v_j)\in\be_n$,}\\
            0 &\text{otherwise.}
        \end{cases}
    \end{align}
\end{definition}

\begin{definition}\label{definition:thresholdgraph}
    A graph $\bg_n = (\bv_n,\be_n)$ with $n$ nodes is called a \textit{threshold graph} if there exists a weight function $\Tresmaprand:\bv_n\rightarrow\mathbb{R}$ and a real number $\tau>0$ such that for any $v_i, v_j\in \bv_n$, including when $i = j$, 
    \begin{align}\label{equation:weight}
    \Tresmaprand(v_i)+\Tresmaprand(v_j) >\tau \textit{\quad if and only if\quad} (v_i,v_j)\in \be_n,
    \end{align}
    where $i,j \in \{1,\ldots,n\}$.
\end{definition}

The next result provides an equivalent characterization of threshold graphs.

\begin{lemma}\label{lemma:thresholdgraph}
    A graph $ \bg_n=(\bv_n,\be_n) $ is a \textit{threshold graph} if and only if there exists a binary $n$-tuple $z^n= (z_1,\ldots, z_n)$, where $z_t \in\{0,1\}$ for all $t\in \{1,\ldots,n\}$, such that  $\bg_n$ can be constructed sequentially, up to node relabeling, using $z^n$ from an empty graph at time step $t=0$ by repeatedly adding a node $v_t$ at time steps $t\in\{1,\ldots,n\}$ that either connects to all existing nodes (i.e., nodes $v_1,\ldots, v_{t-1}$) and itself when $ z_t=1 $ (called a \textit{universal node}), or connects to no existing nodes (including itself) when $ z_t=0$ (called an \textit{isolated node}).

\end{lemma}

Note that the above lemma is equivalent to stating that the resulting edge set $\be^*_t$ at time $ t \in \{1,\ldots, n\}$ satisfies 
    \begin{align}\label{edge-set-eq}
        \be_t^* = 
        \begin{cases}
            \be^*_{t-1} &\textit{if $z_t = 0$}, \\
            \be^*_{t-1}\cup\{(v_t, v_1),\ldots,(v_t,v_t)\}&\textit{if $z_t =1$},
        \end{cases}
    \end{align}
    with $\be_0^*=\emptyset$ and $\be_n^*=\be_n$.   

\begin{proof}[Proof of Lemma~\ref{lemma:thresholdgraph}]
\hfill

\textbf{Forward part:} Consider a threshold graph $\bg_n = (\bv_n, \be_n)$. By definition, this implies that $\bg_n$ is equipped with a weight function $\Tresmaprand$ and a number $\tau>0$ such that $ (v_i,v_j) \in \be_n$ if and only if $ \Tresmaprand(v_i)+\Tresmaprand(v_j)>\tau$.

Let $\bv_n^* = \{v_1^*,\ldots,v_n^*\}$ be the node set obtained by relabeling $\bv_n$ in non-decreasing order of weight, i.e., 
\[
\Tresmaprand(v_1^*)\leq\Tresmaprand(v_2^*)\leq\ldots\leq\Tresmaprand(v_n^*).
\]
To prove the forward part, we construct the binary indicator sequence $z^n= (z_1,\ldots,z_n)$ as follows.
\begin{itemize}
    \item[(1)] At each step, choose two nodes with smallest and largest weight, denoted $v_1^*$ and $v_n^*$ in the first step.
    \item[(2)] If $\Tresmaprand(v_1^*)+\Tresmaprand(v_n^*) \leq \tau$, declare $v_1^*$ as an isolated node. Construct $z^n$ backwards, starting from $z_n$ by setting $z_n= 0$, and define $\bv_{n-1}^* = \bv_n^*\setminus \{v_1^*\}$.
    \item[(3)] If $\Tresmaprand(v_1^*)+\Tresmaprand(v_n^*) > \tau$, declare $v_n^*$ as a universal node setting $z_n= 1$, and define $\bv_{n-1}^* = \bv_n^*\setminus \{v_n^*\}$.
    \item At iteration $j\in\{1,\ldots,n\}$, apply steps $(1)-(3)$ using $\bv^*_{n-j+1}$ until there is no node left (i.e., until $\bv_0^*=\emptyset$ is obtained from $\bv_1^*$ at the end of step $j = n$).
\end{itemize}
The iterative procedure above generates the entire binary indicator vector $z^n= (z_1,\ldots, z_n)$.

Starting with an empty graph $\bg_0^* = (\bv^*_0, \be_0^*)=(\emptyset,\emptyset)$, and repeatedly adding nodes to $\bv^*_0$ and edges to $\be^*_0$ based on $z^n$ and \eqref{edge-set-eq} for $t\in\{1,\ldots,n\}$, in this way we obtain $\be^*_n=\be_n$, as desired.

\textbf{Converse part:} Suppose the graph $\bg_n = (\bv_n, \be_n)$ is a graph which is constructed according to~\eqref{edge-set-eq} via the $n$-tuple $z^n=(z_1,\ldots,z_n)$ of indicator variables. Consider the following weight function on the nodes:
\[
    \Tresmaprand(v_i) = \begin{cases}
         \frac{\tau}{2} + \epsilon_i &\text{\quad if $z_i = 1$},\\
        \frac{\tau}{2} - \epsilon_i &\text{\quad if $z_i = 0$},\\
    \end{cases}\\
\]
where $0<\epsilon_1 < \epsilon_2<\cdots<\epsilon_n<\frac{\tau}{2}$, and $i=1,\ldots,n$.

Consider any pair of nodes $v_i, v_j \in \bv_n$. First suppose that $i=j$. If $z_i=0$, then $(v_i,v_i) \notin \be_n$ and $\Tresmaprand(v_i)+\Tresmaprand(v_i)=\tau-2\epsilon_i<\tau $, and  \eqref{equation:weight} holds. If $ z_i=1 $, then  $\Tresmaprand(v_i)+\Tresmaprand(v_i)=\tau+2\epsilon_i >\tau  $, and \eqref{equation:weight} holds. From this point on, assume without loss of generality that $ i>j $. If $ z_i=0$, then $ (v_i,v_j) \notin \be_n $. Hence $\Tresmaprand(v_i)+\Tresmaprand(v_j)\leq \tau-\epsilon_i+\epsilon_j<\tau $, since $ \epsilon_j <\epsilon_i $, and \eqref{equation:weight} holds. Finally, if $ z_i=1$, then $ (v_i,v_j) \in \be_n $. Hence
$\Tresmaprand(v_i)+\Tresmaprand(v_j)\geq \tau+\epsilon_i-\epsilon_j>\tau $, since $ \epsilon_j <\epsilon_i $, and \eqref{equation:weight} holds. 

\end{proof}

We provide an example next.

\begin{example}
{\em
    Let $\bg_5 = (\bv_5,\be_5) $ be as shown in Figure~\ref{examplegraph1}. 
    \begin{figure}[htbp]
    \centering
    \begin{tikzpicture}
        \node[main node] (1) {$v_1$};
        \node[main node] (2) [left of=1]  {$v_2$};
        \node[main node] (3) [below of = 2] {$v_3$}; 
        \node[main node] (4) [below  of= 1] {$v_4$};
        \node[main node] (5) [below right of= 1] {$v_5$};

        \path[draw,-]
        (1) edge  (2)
        (2) edge  (3)
        (2) edge  (4)
        (1) edge[loop right] (1)
        (2) edge[loop left] (2);
    \end{tikzpicture}
    \caption{A threshold graph of size $5$.}\label{examplegraph1}
    \end{figure}
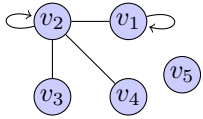
    
    Fix the threshold $\tau = 1$, and consider the weight function $\Tresmaprand$ given by 
    \[
    \begin{cases}
        &\Phi(v_1) = 0.6,\\
        &\Phi(v_2) = 0.8,\\
        &\Phi(v_3) = 0.4,\\
        &\Phi(v_4) = 0.4,\\
        &\Phi(v_5) = 0.1.
    \end{cases}
    \]   
    Since
    $
    \Tresmaprand(v_5)<\Tresmaprand(v_3)=\Tresmaprand(v_4)<\Tresmaprand(v_1)<\Tresmaprand(v_2),
    $
    we obtain the relabeled node set $\bv_n^*=\{v_1^*,v_2^*,v_3^*,v_4^*,v_5^*\}$, where $v_1^* = v_5, v_2^* =v_3, v_3^* = v_4, v_4^* = v_1$ and $v_5^* =v_2$.
    \begin{itemize}
        \item  Now for the first iteration, since $\Tresmaprand(v_1^*)+\Tresmaprand(v_5^*)=0.9<1$, we set $z_5 = 0$, remove $v_1^*$ from $\bv_5^*$ and get $\bv_4^*=\{v_2^*,v_3^*,v_4^*,v_5^*\}$.
        \item For the second iteration, considering the lowest and highest weight nodes in $\bv_4^*$, $v_2^*$ and $v_5^*$, respectively, we have that $\Tresmaprand(v_2^*) +\Tresmaprand(v_5^*)=1.2>1$. Hence we set $z_4 = 1$ and remove $v_5^*$ from $\bv_4^*$, to obtain $\bv_3^*=\{v_2^*,v_3^*,v_4^*\}$.
        \item For the third iteration, taking $v_2^*$ and $v_4^*$ from $\bv_3^*$, we have that $\Tresmaprand(v_2^*)+\Tresmaprand(v_4^*) = 1$. Hence we set $z_3=0$ and remove $v_2^*$ from $\bv_3^*$, which results in the set $\bv_2^*=\{v_3^*,v_4^*\}$. 
        \item For the fourth iteration, we have $\Tresmaprand(v_3^*)+\Tresmaprand(v_4^*) = 1$. Thus we set $z_2 = 0$, and remove $v_3^*$ from $\bv_2^*$, to obtain $\bv_1^*=\{v_4^*\}$.
        \item Finally, since $\Tresmaprand(v_4^*)+\Tresmaprand(v_4^*)=1.2>1$, we set $z_1 = 1$, and remove node $v_4^*$. As no nodes are left (i.e., $\bv_0^*=\emptyset$), we stop the iterative process.
    \end{itemize}
    After five iterations, we have $z^5 = (z_1,z_2,z_3,z_4, z_5) =(1,0,0,1,0)$, and the threshold graph constructed via $z^5$ and \eqref{edge-set-eq} is identical to the original graph $\bg_5$ of Figure~\ref{examplegraph1}.
}
\end{example}

In this paper, we deal with threshold graphs which are generated at random. More specifically, we say a graph $\G_n=(\V_n,\E_n)$ is a \emph{random graph} if its nodes and edges connecting them are generated according to a random process.
Moreover, lowercase $v$ is used for nodes of deterministic graphs $\bg_n=(\bv_n,\be_n)$, and uppercase $V$ is used for nodes of random graphs $\G_n=(\V_n,\E_n)$.
We highlight here the different notation which we consistently use throughout the paper to distinguish between deterministic graphs and random graphs. And in our work here, we say that a random graph $\G_n=(\V_n,\E_n)$ is a {\em random threshold graph} if its random node set $\V_n$ and edge set $\E_n$ satisfy Definition~\ref{definition:thresholdgraph}.

Lemma~\ref{lemma:thresholdgraph} shows that one can equivalently define a random threshold graph of size $n$ by starting with an empty graph and repeatedly adding an isolated node or a universal node according to \eqref{edge-set-eq}.

We next provide a counterpart to Lemma~\ref{lemma:thresholdgraph} for random threshold graphs, whose proof we omit, as it is nearly identical to the proof of Lemma~\ref{lemma:thresholdgraph}. 
\begin{lemma}\label{rand-lemma:thresholdgraph}
    A random graph $ \G_n=(\V_n,\E_n)$ with $n$ nodes is a random threshold graph if and only if there exists a binary $n$-tuple of indicator random variables $Z^n=(Z_1,\ldots,Z_n)$ such that $\G_n$ can be constructed, up to node relabeling, using $Z^n$ from an empty graph at time step $t = 0$ by repeatedly adding a node $V_t$ at time steps $t\in\{1,\ldots,n\}$ that either connects to all existing nodes (i.e., nodes $V_1,\ldots, V_{t-1}$) and itself when $Z_t=1$ (called a universal node), or connects to no existing nodes (including itself) when $Z_t=0$ (called an isolated node). 
\end{lemma}
Lemma~\ref{rand-lemma:thresholdgraph} implies that the evolving edge set $\E_t^*$ at time $t\in\{1,\ldots, n\}$ satisfies \eqref{edge-set-eq} and $\E_n^*=\E_n$.
It is important to note that the notion of random threshold graph and its characterization above are independent of the random process used to generate the indicator random variables. That said, the underlying stochastic process determines the properties of the generated random graphs. In the next section, we specify the stochastic process which we are interested in.

\section{P\'olya threshold graph model and properties}\label{section:properties}
We introduce the notion of a P\'olya threshold graph and study its stochastic and algebraic properties. Throughout, for each $i\in\{1,\ldots,n\}$, the notation $V_i$ denotes the node added to the graph at time step $i$.

As stated, we use the two-color P\'olya urn model as a way to generate random threshold graphs in a sequential manner. 
Starting with an empty graph and a two-color P\'olya urn, at time steps $i\geq 1$, we draw one ball from the P\'olya urn and obtain the indicator draw variable $Z_i\in\{0,1\}$ using \eqref{equation:draw}, i.e., $Z_i=1$ if $i$-th draw is red and $Z_i=0$ if $i$-th draw is black. Given the observed $Z_i$, we classify the node type, add the new node $V_i$ and edges to the graph according to the threshold rule \eqref{edge-set-eq}, i.e., if $Z_i = 1$, then $V_i$ is a \textit{universal node}, if $Z_i= 0$, then $V_i$ is an \textit{isolated node}.
 We repeat this draw process until we obtain a random threshold graph $ \G_n=(\V_n,\E_n)$ of size $n$.
We call the random threshold graph generated in this fashion a \textit{P\'olya threshold graph}.

To describe the adjacency matrix of our threshold graph, we observe that the edge between any pair of nodes depends only on the draw indicator of the more recently added node (i.e., the node with the largest time index). 
Therefore, the adjacency matrix $A_n = [a_{ij}], i,j\in\{1,\ldots,n\}$ of the P\'olya threshold graph is the random variable given by:
\begin{align}
    A_n = [a_{ij}]=[Z_{ \max(i,j)}]=
        \begin{bmatrix}
            Z_1 &Z_2 &Z_3&\cdots&Z_n\\
            Z_2&Z_2&Z_3 &\cdots&Z_n\\
            Z_3&Z_3&Z_3&\cdots&Z_n\\
            \vdots&\vdots&\vdots&\ddots&\vdots\\
            Z_n&Z_n&Z_n&\cdots&Z_n\\
        \end{bmatrix}.\label{eq:polyathreshold-adjacency}
\end{align}
The trace of the adjacency matrix is given as
\[
\mathrm{Tr}(A_n) = \sum_{i = 1}^na_{ii}.
\]
Furthermore, for each node $V_i$ in the graph, its degree $\deg(V_i)$, which counts the number of edges connected to $V_i$ including a potential self-loop, is a random variable determined by draw variables $Z_i,\ldots, Z_n$, as follows:
\begin{align}\label{equation:degreeofVi}
    \deg(V_i) = iZ_i + Z_{i+1}+\ldots+Z_n = iZ_i +\sum_{j=i+1}^nZ_j, 
\end{align}
where $i \in\{1,\ldots, n\}$. 

We begin our analysis by quantifying the stochastic properties of the P\'olya threshold graph induced by the urn reinforcement. A fundamental quantity of interest in any random graph model is the degree distribution of its nodes, as it directly dictates the network's connectivity. As shown in \eqref{equation:degreeofVi}, the degree of any node $V_i$ relies on both its own draw indicator $Z_i$ and the partial sum of the subsequent draw variables. Consequently, the exchangeability and the Beta-Binomial structure of the P\'olya urn process play a central role in our analysis. To derive the exact degree distribution, we first recall the standard definition of the Beta function and highlight the Beta-Binomial distribution exhibited by the indicator draw variables.

\begin{definition}[Beta Function]
The \textit{Beta function} $\map{\bb}{\real_{>0}\times\real_{>0}}{\real_{>0}}$ is given by:
    \[
    \bb(x,y)=\frac{\Gamma(x)\Gamma(y)}{\Gamma(x+y)}.
    \]
\end{definition}

In light of~\eqref{jointProbofPolya}, we have that
\begin{align}
P\bigg(\sum_{i=1}^nZ_i=k\bigg) &= \binom{n}{k}\frac{\bb(\frac{\rho}{\delta}+k,\frac{1-\rho}{\delta}+n-k)}{\bb(\frac{\rho}{\delta},\frac{1-\rho}{\delta})}, \qquad k\in \{0,\ldots, n\}, \label{beta-binom}
\end{align}
where $\rho = \frac{R}{R+B}$ and $\delta = \frac{\Delta}{R+B}$, see also~\autocite{helfand2013polya}. The right-hand side of~\eqref{beta-binom} is nothing but the Beta-Binomial distribution with parameters $n$, $\frac{\rho}{\delta}$ and $\frac{1-\rho}{\delta}$.

Moreover, since the normalized trace of the adjacency matrix satisfies $\frac{1}{n}\mathrm{Tr}(A_n)=\frac{1}{n}\sum_{i=1}^nZ_i$, it follows from Section~\ref{polya-urn} that $\frac{1}{n}\mathrm{Tr}(A_n)$ converges almost surely as $n\rightarrow\infty$ to a Beta random variable $\mathrm{Beta}\big(\frac{\rho}{\delta}, \frac{1-\rho}{\delta}\big)$. 

We next develop the stochastic properties of the P\'olya threshold graph generated by drawing $n$ times from a P\'olya urn with $T= R+B$ initial balls ($R$ red and $B$ black) and the reinforcement parameter $\Delta>0$. Throughout, we let $Z^n= (Z_1,\ldots,Z_n)$ denote the sequence of draw indicators generated by the P\'olya urn-based on~\eqref{equation:draw}. More specifically, we derive the distribution of the degree of any node $V_i$ of the graph, its mean and variance, as well as its centrality score. We first show that the expected degree is the same for all nodes.

\begin{lemma}
    For $i\in\{1,\ldots,n\}$, the expected degree of node $V_i$ satisfies
\begin{align}
    E(\deg(V_i)) = n\rho,
\end{align}
where $\rho=\frac{R}{R+B}$.
\label{theorem:expectation}
\end{lemma}

\begin{proof}
Using the identity in \eqref{equation:degreeofVi} and the linearity of expectation,
\begin{align}
     E(\deg(V_i)) &= E(iZ_i + Z_{i+1} +\cdots + Z_n)\\
    &=E(iZ_i) + E(Z_{i+1}) + \ldots + E(Z_n)\\
    &=iE(Z_i) + \sum_{j=i+1}^nE(Z_j).
\end{align}
By the exchangeability of the draw process $\{Z_i\}_{t=1}^\infty$, noted in Section~\ref{polya-urn}, we have
\begin{equation}
    P(Z_t = 1) = P(Z_1 = 1),
\end{equation}
for all $t\geq 1$. Since $P(Z_i=1) = \rho$, for all $i\in\{1,\ldots,n\}$, we have
\begin{align*}
        E(\deg(V_i)) &=  i\rho +\sum_{j=i+1}^n\rho
        =  n\rho,
\end{align*}
as claimed.
\end{proof}

We now derive the marginal distribution of $\deg(V_i)$. We first clarify that the range of $ \deg(V_i) $ is the set $ D_{i,n}$ defined as follows
\begin{align}\label{rangeofdegV_i}
    D_{i,n} = \begin{cases}
            \{0,\ldots, n\},& \text{ if } i\leq n-i \Longleftrightarrow 2i\leq n, \\
            \{0,\ldots, n-i\}\cup\{i,\ldots,n\}, & \text{ if } i\geq n-i \Longleftrightarrow 2i\geq n.
        \end{cases}
\end{align}
Indeed, if $Z_i = 1$, then $V_i $ is a universal node and has degree at least $i$, and there are up to $(n-i)$ nodes that may become connected to $V_i$ when generated from time steps $i+1,\ldots,n$; hence the largest value that $ \deg(V_i) $  can take is $ n $, which means that $k\in\{i,\ldots, n\}$. For $Z_i = 0$, $V_i$ is an isolated node and there are up to $n-i$ nodes that may be connected to $V_i$ when added at time steps $i+1,\ldots, n$; hence the largest value that $\deg(V_i)$ can take is $n-i$, which means, we have $k\in\{0,\ldots, n-i\}$. Finally, note that when $2i\leq n$, we have $\{0,\ldots, n-i\}\cup\{i,\ldots,n\} = \{0,\ldots, n\}$. We have the following result.
\begin{theorem} 
For the P\'olya threshold graph, the probability distribution of $\deg(V_i)$ is given by
\begin{align}\label{equation:probability_of_degree}
        P(\deg(V_i) = k) = \binom{n-i}{k}g_i^{(n-i+1)}(k)\mathbb{I}(k\leq n-i) + \binom{n-i}{k-i}g_i^{(n-i+1)}(k-i+1)\mathbb{I}(k\geq i),
\end{align}
for $k\in D_{i,n}$, where $\mathbb{I}(\cdot)$ denotes the indicator function, and
\begin{align}
    g_i^{(n-i+1)}(k) = \frac{\Gamma\big(\frac{1}{\delta}\big)\Gamma\big(\frac{\rho}{\delta}+k\big)\Gamma\big(\frac{1-\rho}{\delta}+n-i+1-k\big)}{\Gamma\big(\frac{\rho}{\delta}\big)\Gamma\big(\frac{1-\rho}{\delta}\big)\Gamma\big(\frac{1}{\delta}+n-i+1\big)},
\end{align}
with $\delta = \frac{\Delta}{T}$.
\label{theorem:degreeprob}
\end{theorem}

We first prove an intermediate result.

\begin{lemma}\label{lemma:degreefunction}
    For all $i\in\{1,\ldots, n\}$ and $k\in D_{i,n}$,
    \begin{equation}\label{equation:degreefunction}
        P(\deg(V_i)=k) = P\Big(Z_i = 1, \sum^{n}_{j = i+1}Z_j = k-i\Big)\mathbb{I}(k\geq i) + P\Big(Z_i = 0, \sum_{j = i+1}^n Z_j = k\Big)\mathbb{I}(k\leq n-i).
    \end{equation}
    
\end{lemma}

\begin{proof}[Proof]
In light of the expression in \eqref{rangeofdegV_i} for the range $D_{i,n}$ of $\deg(V_i)$, we have that if $Z_i = 1$, then $k\in\{i,\ldots, n\}$; if $Z_i = 0$, then $k\in\{0,\ldots,n-i\}$.
    This yields the stated two-term decomposition in \eqref{equation:degreefunction}.

\end{proof}

\begin{proof}[Proof of Theorem~\ref{theorem:degreeprob}]
By Lemma~\ref{lemma:degreefunction} and the exchangeability of the draw process, we have

\begin{align*}
        P(\deg(V_i) = k)  &= P\Big(Z_i = 1, \sum^{n}_{j = i+1}Z_j = k-i\Big)\mathbb{I}(k\geq i) + P\Big(Z_i = 0, \sum_{j = i+1}^n Z_j = k\Big)\mathbb{I}(k\leq n-i)\\
        &= \binom{n-i}{k-i}P(Z_1 = 1,\ldots, Z_{k-i+1} = 1, Z_{k-i+2} = 0,\ldots Z_{n-i+1}= 0)\mathbb{I}(k\geq i)\\*
        & \qquad +\binom{n-i}{k}P(Z_1 = 1,\ldots, Z_k = 1, Z_{k+1} = 0,\ldots Z_{n-i+1} = 0)\mathbb{I}(k\leq n-i).
\end{align*}
Using the joint law of the P\'olya draw random variables given in~\eqref{jointprobofPolya_product}-\eqref{jointProbofPolya} yields
\begin{align}        
        P(\deg(V_i) = k)&=  \binom{n-i}{k-i}\biggl[\frac{\Gamma\big(\frac{\rho}{\delta}+(k-i+1)\big)\Gamma\big(\frac{1-\rho}{\delta}+(n-k)\big)\Gamma\big(\frac{1}{\delta}\big)}{\Gamma\big(\frac{\rho}{\delta}\big)\Gamma\big(\frac{1-\rho}{\delta}\big)\Gamma\big(\frac{1}{\delta}+(n-i+1)\big)}  \biggr]\mathbb{I}(k\geq i)\nonumber\\*
        & \qquad +\binom{n-i}{k} \biggl[  \frac{\Gamma\big(\frac{\rho}{\delta}+k\big)\Gamma\big(\frac{1-\rho}{\delta}+(n-k-i+1)\big)\Gamma\big(\frac{1}{\delta}\big)}{\Gamma\big(\frac{\rho}{\delta}\big)\Gamma\big(\frac{1-\rho}{\delta}\big)\Gamma\big(\frac{1}{\delta}+(n-i+1)\big)}\biggr]\mathbb{I}(k\leq n-i)\\
         &=  \binom{n-i}{k-i}g_i^{(n-i+1)}(k-i+1)\mathbb{I}(k\geq i) + \binom{n-i}{k}g_i^{(n-i+1)}(k)\mathbb{I}(k\leq n-i),\nonumber
\end{align}
where the last equation is obtained by using the fact that $\Gamma(x+1) = x\Gamma(x)$ for any positive real number.
\end{proof}

To prepare for the variance calculation, we first introduce the rising factorial and the Chu-Vandermonde identity.

\begin{definition}[\cite{steffensen2013interpolation}] \label{risingfactorial}
For $x\in\mathbb{R}$ and integer $n\ge 0$, the rising factorial is given by
    \begin{equation}
        (x)^{(n)} = x(x+1)\cdots(x+n-1) = \prod_{k=0}^{n-1}(x+k).
    \end{equation}
\end{definition}

\begin{lemma}[Chu-Vandermonde Identity]\label{chu-vandermonde}\autocite{graham1994concrete}
For $n\in \mathbb{Z}_{\geq 0}$ and $s,t\in\real$,
    \begin{equation}
        (s+t)^{(n)} = \sum_{k=0}^{n} \binom{n}{k} s^{(k)} t^{(n-k)}.
    \end{equation}
\end{lemma}

We next compute the variance of $\deg(V_i)$.
\begin{theorem}
    For each $i\in\{1,\ldots,n\}$, the variance of $\deg(V_i)$ for node $V_i$ is given by
    \begin{align}
        \Var(\deg(V_i))& = \Big(1+\frac{2i}{\frac{1}{\delta}+n-i}\Big)\Bigg[\frac{(n-i)\rho\big[(n-i)(\delta+\rho)+1-\rho\big]}{1+\delta}\Bigg]+i^2\rho \nonumber\\
        &\qquad \qquad + \Big(\frac{2i\rho^2(n-i)}{1+n\delta-i\delta}\Big)-(n\rho)^2.
    \end{align}
    \label{theorem:variance}
\end{theorem}

\begin{proof}[Proof] We use
    \begin{align}
        \Var(\deg(V_i)) &= E[\deg(V_i)^2] - E[\deg(V_i)]^2,
    \end{align}
where the second moment $E(\deg(V_i)^2)$ can be expressed as
\begin{align*}
    E(\deg(V_i)^2) &= \sum_{k=0}^{n-i}k^2\binom{n-i}{k}\Biggl[\frac{\Gamma\big(\frac{\rho}{\delta}+k\big)\Gamma\big(\frac{1-\rho}{\delta}+n-k-i+1\big)\Gamma\big(\frac{1}{\delta}\big)}{\Gamma\big(\frac{\rho}{\delta}\big)\Gamma\big(\frac{1-\rho}{\delta}\big)\Gamma\big(\frac{1}{\delta}+n-i+1\big)}\Biggr] \\*
    & \qquad +\sum_{k=i}^{n}k^2\binom{n-i}{k-i}\Biggl[\frac{\Gamma\big(\frac{\rho}{\delta}+k-i+1\big)\Gamma\big(\frac{1-\rho}{\delta}+n-k\big)\Gamma\big(\frac{1}{\delta}\big)}{\Gamma\big(\frac{\rho}{\delta}\big)\Gamma\big(\frac{1-\rho}{\delta}\big)\Gamma\big(\frac{1}{\delta}+n-i+1\big)}\Biggr].
\end{align*}
Using the change of variables $x=k-i$ in the second sum, we obtain
\begin{align*}
    E(\deg(V_i)^2) &= \sum_{k=0}^{n-i}k^2\binom{n-i}{k}\Biggl[\frac{\Gamma\big(\frac{\rho}{\delta}+k\big)\Gamma\big(\frac{1-\rho}{\delta}+n-k-i+1\big)\Gamma\big(\frac{1}{\delta}\big)}{\Gamma\big(\frac{\rho}{\delta}\big)\Gamma\big(\frac{1-\rho}{\delta}\big)\Gamma\big(\frac{1}{\delta}+n-i+1\big)}\Biggr] \\*
    & \qquad +\sum_{x=0}^{n-i}(x+i)^2\binom{n-i}{x}\Biggl[\frac{\Gamma\big(\frac{\rho}{\delta}+x+1\big)\Gamma\big(\frac{1-\rho}{\delta}+n-x-i\big)\Gamma\big(\frac{1}{\delta}\big)}{\Gamma\big(\frac{\rho}{\delta}\big)\Gamma\big(\frac{1-\rho}{\delta}\big)\Gamma(\frac{1}{\delta}+n-i+1)}\Biggr].
\end{align*}
Expanding $(x+i)^2=x^2+2ix+i^2$ yields
\begin{align*}
    E(\deg(V_i)^2) &= \sum_{k=0}^{n-i}k^2\binom{n-i}{k}\Biggl[\frac{\Gamma\big(\frac{\rho}{\delta}+k\big)\Gamma\big(\frac{1-\rho}{\delta}+n-k-i+1\big)\Gamma\big(\frac{1}{\delta}\big)}{\Gamma\big(\frac{\rho}{\delta}\big)\Gamma\big(\frac{1-\rho}{\delta}\big)\Gamma\big(\frac{1}{\delta}+n-i+1\big)}\Biggr] \\*
    & \qquad +\sum_{x=0}^{n-i}x^2\binom{n-i}{x}\Biggl[\frac{\Gamma\big(\frac{\rho}{\delta}+x+1\big)\Gamma\big(\frac{1-\rho}{\delta}+n-x-i\big)\Gamma\big(\frac{1}{\delta}\big)}{\Gamma\big(\frac{\rho}{\delta}\big)\Gamma\big(\frac{1-\rho}{\delta}\big)\Gamma(\frac{1}{\delta}+n-i+1)}\Biggr]\\*
    & \qquad +\sum_{x=0}^{n-i}i^2\binom{n-i}{x}\Biggl[\frac{\Gamma\big(\frac{\rho}{\delta}+x+1\big)\Gamma\big(\frac{1-\rho}{\delta}+n-x-i\big)\Gamma\big(\frac{1}{\delta}\big)}{\Gamma\big(\frac{\rho}{\delta}\big)\Gamma\big(\frac{1-\rho}{\delta}\big)\Gamma(\frac{1}{\delta}+n-i+1)}\Biggr]\\*
    & \qquad +\sum_{x=0}^{n-i}2ix\binom{n-i}{x}\Biggl[\frac{\Gamma\big(\frac{\rho}{\delta}+x+1\big)\Gamma\big(\frac{1-\rho}{\delta}+n-x-i\big)\Gamma\big(\frac{1}{\delta}\big)}{\Gamma\big(\frac{\rho}{\delta}\big)\Gamma\big(\frac{1-\rho}{\delta}\big)\Gamma(\frac{1}{\delta}+n-i+1)}\Biggr].
\end{align*}
We now group the second terms with the first sum. Recall from~\eqref{beta-binom} that the Beta-Binomial random variable $X$ with parameters $N$, $\alpha$ and $\beta$ has distribution
\begin{align*}
    P(X=k) = \binom{N}{k}\frac{\Gamma(k+\alpha)\Gamma(N-k+\beta)\Gamma(\alpha+\beta)}{\Gamma(\alpha)\Gamma(\beta)\Gamma(N+\alpha+\beta)}, \qquad k \in \{0, 1, \dots, N\},
\end{align*}
and its second moment is given by (see~\autocite{johnson2005univariate}) 
\[E[X^2] = \frac{N\alpha[N(1+\alpha)+\beta]}{(\alpha+\beta)(1+\alpha+\beta)}.
\]
Note that the combined sum coincides with the second moment of a Beta-Binomial random variable with parameters $n-i$, $\frac{\rho}{\delta}$ and $\frac{1-\rho}{\delta}$, which gives
\begin{align*}
    \sum_{k=0}^{n-i}k^2&\binom{n-i}{k}\Biggl[\frac{\Gamma\big(\frac{\rho}{\delta}+k\big)\Gamma\big(\frac{1-\rho}{\delta}+n-k-i\big)\Gamma\big(\frac{1}{\delta}\big)}{\Gamma\big(\frac{\rho}{\delta}\big)\Gamma\big(\frac{1-\rho}{\delta}\big)\Gamma\big(\frac{1}{\delta}+n-i\big)}\Biggr]\\*
    &=\frac{(n-i)\frac{\rho}{\delta}\bigl[(n-i)(1+\frac{\rho}{\delta})+\frac{1-\rho}{\delta}\bigr]}{\frac{1}{\delta}\bigl(\frac{1}{\delta}+1\bigr)}\\*
    &=\frac{(n-i)\rho\bigl[(n-i)(\delta+\rho)+1-\rho\bigr]}{1+\delta}.
\end{align*}

For the $i^2$-term, we obtain
\begin{align}
    \sum_{x=0}^{n-i}i^2\binom{n-i}{x}&\Biggl[\frac{\Gamma\big(\frac{\rho}{\delta}+x+1\big)\Gamma\big(\frac{1-\rho}{\delta}+n-x-i\big)\Gamma\big(\frac{1}{\delta}\big)}{\Gamma\big(\frac{\rho}{\delta}\big)\Gamma\big(\frac{1-\rho}{\delta}\big)\Gamma\big(\frac{1}{\delta}+n-i+1\big)}\Biggr]\nonumber \\*
    &=i^2 \bigg[\sum_{x=0}^{n-i}\binom{n-i}{x}\frac{(\delta)\Gamma(\frac{1}{\delta}+1)}{(\frac{\delta}{\rho})\Gamma(\frac{\rho}{\delta}+1)\Gamma(\frac{1-\rho}{\delta})}\frac{\Gamma(\frac{\rho}{\delta}+x+1)\Gamma(\frac{1-\rho}{\delta}+n-x-i)}{\Gamma(\frac{1}{\delta}+n-i+1)} \bigg]\nonumber\\*
    &= i^2\rho\bigg(\frac{\Gamma(\frac{1}{\delta}+1)}{\Gamma(\frac{\rho}{\delta}+1)\Gamma(\frac{1-\rho}{\delta})}\bigg)\nonumber\\*
    &\qquad\bigg[\sum_{x=0}^{n-i}\binom{n-i}{x}\frac{\Gamma(\frac{\rho}{\delta}+1)\Gamma(\frac{1-\rho}{\delta})\prod_{s=0}^{x-1}(\frac{\rho}{\delta}+1+s)\prod_{j=0}^{n-x-i-1}(\frac{1-\rho}{\delta}+j)}{\Gamma(\frac{1}{\delta}+1)\prod_{k=0}^{n-i-1}(\frac{1}{\delta}+1+k)}\bigg]\nonumber\\*
    &= i^2\rho\bigg[\frac{\sum_{x=0}^{n-i}\binom{n-i}{x}\prod_{s=0}^x(\frac{\rho}{\delta}+1+s)\prod_{j=0}^{n-x-i-1}(\frac{1-\rho}{\delta}+j)}{\prod_{k=0}^{n-i}(\frac{1}{\delta}+1+k)}\bigg]\nonumber\\
    &=  i^2\rho\bigg[\frac{\sum_{x=0}^{n-i}\binom{n-i}{x}(\frac{\rho}{\delta}+1)^{(x)}(\frac{1-\rho}{\delta})^{(n-x-i)}}{(\frac{1}{\delta}+1)^{(n-i)}}\bigg]\label{Use rising factorial}\\
    &= i^2\rho \bigg[\frac{(\frac{\rho}{\delta}+\frac{1-\rho}{\delta}+1)^{(n-i)}}{(\frac{1}{\delta}+1)^{(n-i)}}\bigg]\label{equation:useLemma5}\\*
    &=i^2\rho,\nonumber
\end{align}
where \eqref{Use rising factorial} follows from the definition of \textit{rising factorial} and \eqref{equation:useLemma5} is obtained by using Lemma~\ref{chu-vandermonde}.

Therefore,
\begin{align}\label{equation:second moment}
    E(\deg(V_i)^2) &=\frac{(n-i)\rho\bigl[(n-i)(\delta+\rho)+1-\rho\bigr]}{1+\delta}+i^2\rho \nonumber\\*
    & \quad +\sum_{x=0}^{n-i}2ix\binom{n-i}{x}\Bigl[\frac{\Gamma\big(\frac{\rho}{\delta}+x+1\big)\Gamma\big(\frac{1-\rho}{\delta}+n-x-i\big)\Gamma\big(\frac{1}{\delta}\big)}{\Gamma\big(\frac{\rho}{\delta}\big)\Gamma\big(\frac{1-\rho}{\delta}\big)\Gamma\big(\frac{1}{\delta}+n-i+1\big)}\Bigr].
\end{align}

The last term in \eqref{equation:second moment} can be simplified as follows
\begin{align}\label{equation:xandx^2}
\sum_{x=0}^{n-i}2ix&\binom{n-i}{x}\biggl[\frac{\Gamma\big(\frac{\rho}{\delta}+x+1\big)\Gamma\big(\frac{1-\rho}{\delta}+n-x-i\big)\Gamma\big(\frac{1}{\delta}\big)}{\Gamma\big(\frac{\rho}{\delta}\big)\Gamma\big(\frac{1-\rho}{\delta}\big)\Gamma\big(\frac{1}{\delta}+n-i+1\big)}\Biggr] \nonumber\\
&=\sum_{x=0}^{n-i}2ix\frac{(\frac{\rho}{\delta}+x)}{(\frac{1}{\delta}+n-i)}\binom{n-i}{x}\Biggl[\frac{\Gamma\big(\frac{\rho}{\delta}+x\big)\Gamma\big(\frac{1-\rho}{\delta}+n-x-i\big)\Gamma\big(\frac{1}{\delta}\big)}{\Gamma\big(\frac{\rho}{\delta}\big)\Gamma\big(\frac{1-\rho}{\delta}\big)\Gamma\big(\frac{1}{\delta}+n-i\big)}\Biggr]\nonumber\\
&=  \sum_{x=0}^{n-i}\frac{2i\frac{\rho}{\delta}x}{(\frac{1}{\delta}+n-i)}\binom{n-i}{x}\Biggl[\frac{\Gamma\big(\frac{\rho}{\delta}+x\big)\Gamma\big(\frac{1-\rho}{\delta}+n-x-i\big)\Gamma\big(\frac{1}{\delta}\big)}{\Gamma\big(\frac{\rho}{\delta}\big)\Gamma\big(\frac{1-\rho}{\delta}\big)\Gamma\big(\frac{1}{\delta}+n-i\big)}\Biggr] \nonumber\\*
    &\qquad +\sum_{x=0}^{n-i}\frac{2ix^2}{(\frac{1}{\delta}+n-i)}\binom{n-i}{x}\Biggl[\frac{\Gamma\big(\frac{\rho}{\delta}+x\big)\Gamma\big(\frac{1-\rho}{\delta}+n-x-i\big)\Gamma\big(\frac{1}{\delta}\big)}{\Gamma\big(\frac{\rho}{\delta}\big)\Gamma\big(\frac{1-\rho}{\delta}\big)\Gamma\big(\frac{1}{\delta}+n-i\big)}\Biggr].
\end{align}
By the first and second moments of the Beta-Binomial distribution with parameters $n-i$, $\frac{\rho}{\delta}$ and $\frac{1-\rho}{\delta}$, we obtain
\begin{align*}
\sum_{x=0}^{n-i}\frac{2i\frac{\rho}{\delta}x}{(\frac{1}{\delta}+n-i)}\binom{n-i}{x}&\Biggl[\frac{\Gamma\big(\frac{\rho}{\delta}+x\big)\Gamma\big(\frac{1-\rho}{\delta}+n-x-i\big)\Gamma\big(\frac{1}{\delta}\big)}{\Gamma\big(\frac{\rho}{\delta}\big)\Gamma\big(\frac{1-\rho}{\delta}\big)\Gamma\big(\frac{1}{\delta}+n-i\big)}\Biggr]=
\frac{2i\frac{\rho}{\delta}}{(\frac{1}{\delta}+n-i)}\Big(\frac{(n-i)\frac{\rho}{\delta}}{\frac{1}{\delta}}\Big),\\
\sum_{x=0}^{n-i}\frac{2ix^2}{(\frac{1}{\delta}+n-i)}\binom{n-i}{x}&\Biggl[\frac{\Gamma\big(\frac{\rho}{\delta}+x\big)\Gamma\big(\frac{1-\rho}{\delta}+n-x-i\big)\Gamma\big(\frac{1}{\delta}\big)}{\Gamma\big(\frac{\rho}{\delta}\big)\Gamma\big(\frac{1-\rho}{\delta}\big)\Gamma\big(\frac{1}{\delta}+n-i\big)}\Biggr] = \\
&\frac{2i}{(\frac{1}{\delta}+n-i)}\Bigg[\frac{(n-i)\frac{\rho}{\delta}\big[(n-i)(1+\frac{\rho}{\delta})+\frac{1-\rho}{\delta}\big]}{\frac{1}{\delta}\big(1+\frac{1}{\delta}\big)}\Bigg].
\end{align*}
Replacing those into \eqref{equation:xandx^2}, we obtain
\begin{align*}
     E(\deg(V_i)^2) &=\frac{(n-i)\rho\bigl[(n-i)(\delta+\rho)+1-\rho\bigr]}{1+\delta}+i^2\rho \\*
     &\quad+\frac{2i\frac{\rho}{\delta}}{(\frac{1}{\delta}+n-i)}\Big(\frac{(n-i)\frac{\rho}{\delta}}{\frac{1}{\delta}}\Big) + \frac{2i}{(\frac{1}{\delta}+n-i)}\Bigg[\frac{(n-i)\frac{\rho}{\delta}\big[(n-i)(1+\frac{\rho}{\delta})+\frac{1-\rho}{\delta}\big]}{\frac{1}{\delta}\big(1+\frac{1}{\delta}\big)}\Bigg]\\
    &= i^2\rho +\frac{\rho(n-i)\big[(n-i)(\delta+\rho)+1-\rho\big]}{1+\delta}\\*
    & \qquad + \frac{2i\frac{\rho^2}{\delta}(n-i)}{(\frac{1}{\delta}+n-i)} + \frac{2i}{(\frac{1}{\delta}+n-i)}\frac{(n-i)\rho\big[(n-i)(\delta+\rho)+1-\rho\big]}{1+\delta}\\
    &= \Bigg(1+\frac{2i}{\frac{1}{\delta}+n-i}\Bigg)\Bigg[\frac{(n-i)\rho\big[(n-i)(\delta+\rho)+1-\rho\big]}{1+\delta}\Bigg]+i^2\rho + \frac{2i\rho^2(n-i)}{1+n\delta-i\delta},
\end{align*}
as desired.
\end{proof}

While the degree distribution precisely characterizes a node's local connectivity, it does not fully describe its global position within the network. To capture a node's influence beyond local degree statistics, we consider the decay centrality score, see~\autocite{tsakas2017decaycentrality}. Widely used to analyze network efficiency, this metric evaluates network influence by discounting low-degree nodes. We next evaluate the centrality score of the P\'olya threshold graph.

\begin{definition} [\autocite{tsakas2017decaycentrality}]
    The centrality score of an arbitrary node $V_i$ in a graph of size $n$ is defined as
    \begin{align}
        C_{V_i} = \sum_{j=1}^n\alpha^{d(V_i,V_j)},
    \end{align}
    where $\alpha \in(0,1)$.
\end{definition}
In this paper, we adopt the standard choice of $\alpha = \frac{1}{2}$.

Next, we denote $d(V_i, V_j)$ as the shortest distance between nodes $V_i$ and $V_j$. We adopt the convention that a node without a self-loop is treated as disconnected from itself, i.e., $d(V_i,V_i)=\infty$ when $Z_i=0$, see for instance,~\autocite{BlochJacksonTebaldi2023Centrality}. 

For the P\'olya threshold graph, the distance between nodes $V_i$ and $V_j$ is given by
\begin{align}\label{equation:distance}
    d(V_i,V_j) = \begin{cases}
            0 & \text{ if $Z_i = 1, i = j$, }\\
            \infty & \text{ if $Z_i = 0, i=j$,}\\
            1 & \text{ if $Z_{ \max(i,j)} = 1, i\neq j$,}\\
            2 & \text{ if $Z_{ \max(i,j)}=0$ and $\sum^n_{k= \max(i,j)+1}Z_k \geq 1, i\neq j$,}\\
            \infty & \text{if $\sum^n_{k= \max(i,j)}Z_k = 0, i\neq j$.} 
    \end{cases}
\end{align}

\begin{lemma}\label{lemma:prob of distance}
    The probability distribution of the distance $d(V_i,V_j)$ in \eqref{equation:distance} is given by
    \begin{align*}
        &P(d(V_i, V_i)= 0)= \rho,\\*
        &P(d(V_i, V_i) =\infty)  = 1-\rho,\\*
        &P(d(V_i,V_j) = 1 ) = \rho,  \quad i\neq j,\\
        &P(d(V_i,V_j) = 2) =  1-\Big[ \rho+\Big(\prod_{s=0}^{n-{\max(i,j)}}\frac{1-\rho+s\delta}{1+s\delta}\Big)\Big],\quad i\neq j,\\*
        &P(d(V_i,V_j) = \infty) = \prod_{s=0}^{n-{\max(i,j)}}\frac{1-\rho+s\delta}{1+s\delta},\quad i\neq j.
    \end{align*}
\end{lemma}
\begin{proof}[Proof]
    First, for $i=j$, $d(V_i,V_i) = 0$ if and only if $Z_i=1$ (a self--loop is present), and $d(V_i,V_i) = \infty$ otherwise, which gives
    \begin{align*}
        &P(d(V_i, V_i)= 0)= P(Z_i = 1) = \rho,\\*
        &P(d(V_i, V_i) =\infty) =P(Z_i = 0) = 1-\rho.
    \end{align*}

    Without loss of generality, let $i<j$. If $Z_j=1$, then $V_j$ is a universal node and directly connects to $V_i$, so $d(V_i,V_j)=1$, where
    \[
    P(d(V_i,V_j)=1) = P(Z_j = 1) = \rho.
    \]
    If $Z_j = 0$, then $V_j$ is isolated at time $j$. If some later node is universal, it connects to both $V_i$ and $V_j$, giving distance $2$; if no later node is universal, there is no path connecting $V_i$ and $V_j$, and the distance is $\infty$. The stated probabilities follow from exchangeability and the P\'olya joint law of the draw random variables:
    \begin{align*}
        P(d(V_i,V_j) = 2) &= P\bigg(Z_{j}=0,\sum_{k={j+1}}^nZ_k \ge 1\bigg)\\
        &= P(Z_{j}=0) - P\bigg(Z_{j}=0,\sum_{k={j+1}}^nZ_k =0\bigg)\\
        &= (1-\rho)-\Bigg[ \prod_{s=0}^{n-{j}}\frac{1-\rho+s\delta}{1+s\delta}\Bigg].
    \end{align*}
    Finally, 
    \begin{align*}
        P(d(V_i, V_j) &= \infty) = P\bigg(\sum_{k=j}^{n}Z_k = 0\bigg) \\*
        &=(1-\rho)\frac{1-\rho+\delta}{1+\delta}\cdots\frac{1-\rho+(n-j)\delta}{1+(n-j)\delta}\\*
        &= \prod_{s=0}^{n-j}\frac{1-\rho+s\delta}{1+s\delta},
    \end{align*}
    as claimed.
\end{proof}

Using these distance probabilities, we now derive the expected centrality score.

\begin{theorem}\label{theorem:centralityscore}
    For each $i\in\{1,\ldots,n\}$,
    \begin{align}\label{equation:centralityscore}
        E(C_{V_i})&= \sum_{j=1, j\neq i}^n\Bigg[\frac{\rho}{4}+\frac{1}{4}\bigg(1-\prod_{s=0}^{n-{\max(i,j)}}\frac{1-\rho+s\delta}{1+s\delta}\bigg)\Bigg]+\rho.
    \end{align}
\end{theorem}

\begin{proof}
Note that
    \begin{align}\label{centrality_proof_equation_1}
        E(C_{V_i}) &= E\bigg(\sum_{j=1}^n \frac{1}{2^{d(V_i,V_j)}}\bigg)\nonumber\\
        & = E\bigg(\sum_{j=1, j\neq i}^n \frac{1}{2^{d(V_i,V_j)}}\bigg)+E\bigg(\frac{1}{2^{d(V_i,V_i)}}\bigg).
    \end{align}
Since $d(V_i,V_i)=0$ when $Z_i=1$ and $d(V_i,V_i)=\infty$ when $Z_i=0$,
    \begin{align*}
        E\bigg(\frac{1}{2^{d(V_i,V_i)}}\bigg) = \bigg[\frac{1}{2^0}P(Z_i = 1)+\lim_{s\rightarrow \infty}\frac{1}{2^s}P(Z_i = 0)\bigg] = \rho.
    \end{align*}
    And the first term in \eqref{centrality_proof_equation_1} can be shown as
    \begin{align*}
        E\bigg(\sum_{j=1, j\neq i}^n \frac{1}{2^{d(V_i,V_j)}}\bigg)&= \sum_{j=1, j\neq i}^n\bigg[\frac{1}{2^1}P(Z_{ \max(i,j)}=1)+\lim_{s\rightarrow \infty}\frac{1}{2^s}P\bigg(\sum_{k= \max(i,j)}^nZ_k = 0\bigg)\\*
        & \qquad+\frac{1}{2^2}P\bigg(Z_{ \max(i,j)} = 0, \sum_{k=\max(i,j)+1}^nZ_k\geq1\bigg)\Bigg]\\
        &= \sum_{j=1, j\neq i}^n\Bigg[\frac{\rho}{2}+0+\frac{1}{4}\bigg(1-\rho- P\bigg(\sum_{k= \max(i,j)}^nZ_k = 0\bigg)\bigg)\Bigg]\\
        &= \sum_{j=1, j\neq i}^n\Bigg[\frac{\rho}{2}+\frac{1}{4}\bigg(1-\rho-\prod_{s=0}^{n-{\max(i,j)}}\frac{1-\rho+s\delta}{1+s\delta}\bigg)\Bigg]\\
        &= \sum_{j=1, j\neq i}^n\Bigg[\frac{\rho}{4}+\frac{1}{4}\bigg(1-\prod_{s=0}^{n-{\max(i,j)}}\frac{1-\rho+s\delta}{1+s\delta}\bigg)\Bigg],
    \end{align*}
and substituting the probabilities in Lemma~\ref{lemma:prob of distance} above yields \eqref{equation:centralityscore}.
\end{proof}

Having established the stochastic properties of the P\'olya threshold graph, including the degree distribution and centrality scores, we then conclude this section by characterizing some of its algebraic properties. As a fundamental algebraic representation of a network, the Laplacian matrix captures the deep structural connectivity of the network, and provides the essential mathematical foundation for studying linear processes over the graph~\autocite{CG-RG:01}. To this end, we first formally define the degree and Laplacian matrices of the P\'olya threshold graph.
\begin{definition}[\textbf{Degree and Laplacian matrices}]
    The degree matrix $D_n$ of $\G_n$ is the random variable given by the diagonal matrix whose $i$-th diagonal entry is $\deg(V_i)$:
    \[
        D_n := \mathrm{diag}\big(\deg(V_1),\ldots,\deg(V_n)\big).
    \]
    The Laplacian matrix of the P\'olya threshold graph is the random variable defined by
    \begin{align}\label{eq:Ln-explicit}
        L_n := D_n - A_n.
    \end{align}
    Hence, in light of \eqref{eq:polyathreshold-adjacency} and \eqref{equation:degreeofVi}, the entries of $L_n$ satisfy
    \begin{equation}\label{eq:Ln-entries}
        (L_n)_{ij}
        =
        \begin{cases}
            \deg(V_i)-Z_i,  & \text{if } i = j,\\[0.2em]
            -Z_{\max(i,j)}, & \text{if } i\neq j,
        \end{cases}
    \end{equation}
    for $i,j \in \{1,\ldots, n\}$.
\end{definition}

Next, we describe the spectrum of $L_n$.  
While the Laplacian spectrum of a deterministic threshold graph is known to be determined by its degree sequence~\autocite{merris1994degree, GroneMerris1994, Helmberg2015}, the eigenvalues are not, in general, equal to the degrees. By contrast, for our P\'olya threshold graph, the nonzero Laplacian eigenvalues coincide with the degrees of the corresponding nodes, beginning with the second node. We now establish this result and derive a corresponding eigenbasis.
Let $\{e_i\}_{i=1}^n$ denote the standard basis of $\mathbb{R}^n$.

\begin{theorem}\label{theorem:eigenvalue}
    Let $L_n$ be the Laplacian matrix of a P\'olya threshold graph with $n$ nodes. Then the multiset of eigenvalues of $L_n$, $\Lambda(L_n)$, is given by 
    \begin{align}
        \Lambda(L_n) = \{0,\deg(V_2),\deg(V_3),\ldots,\deg(V_n)\}
        = \Big\{0,\; 2Z_2+\!\!\sum_{j=3}^n Z_j,\; \ldots,\; nZ_n\Big\}.
    \end{align}
    A corresponding eigenbasis of $\mathbb{R}^n$ is given by
    \begin{align}
        u_1 &= (1,1,\ldots,1)^\top \in \mathbb{R}^n,\\
        u_m &= \sum_{i=1}^{m-1} e_i - (m-1)e_m
             = (\,\underbrace{1,\ldots,1}_{m-1},\,-(m-1),\,\underbrace{0,\ldots,0}_{n-m}\,)^\top,
    \end{align}
    where $m=2,\ldots,n$.
\end{theorem}

\begin{proof}
    First, observe that for every $i$, the sum of the entries in the $i$-th row of $A_n$ equals $\deg(V_i)$, and $a_{ii}=Z_i$.  
    Hence, the $i$-th entry of $L_n u_1$ is
    \[
        (L_n u_1)_i
        = \big(\deg(V_i)-Z_i\big) - \sum_{j\neq i} Z_{\max(i,j)}
        = 0,
    \]
    where we have used \eqref{equation:degreeofVi}. Thus, $u_1$ is an eigenvector with eigenvalue $0$.

    Fix $m\in\{2,\ldots,n\}$ and consider
    \[
        u_m = \sum_{i=1}^{m-1} e_i - (m-1)e_m.
    \]
    For each $i\in\{1,\dots,n\}$,
    \[
    (L_nu_m)_i=\sum_{j=1}^n (L_n)_{ij}(u_m)_j.
    \]
    The only nonzero coordinates of $u_m$ are $(u_m)_k=1$ for $1\le k\le m-1$ and $(u_m)_m=-(m-1)$, with $(u_m)_k=0$ for $k>m$.  
    Using the explicit form of $L_n$ in \eqref{eq:Ln-explicit}-\eqref{eq:Ln-entries}, we have
    \begin{align}\label{equation:Laplacian and eigenvector}
        (L_n u_m)_i =
        \begin{cases}
             (i-1)(-Z_i)+ (\deg(V_i)-Z_i)-\sum_{j=i+1}^{m-1}Z_j-(m-1)(-Z_m), & 1\le i\le m-1,\\[0.2em]
            (m-1)(-Z_m)-(m-1)\big[\deg(V_m)-Z_m\big], & i=m,\\[0.2em]
            (m-1)(-Z_i)-(m-1)(-Z_i), & i>m.
        \end{cases}
    \end{align}
    Substituting \eqref{equation:degreeofVi} into \eqref{equation:Laplacian and eigenvector} yields 
    \[
        (L_n u_m)_i =
        \begin{cases}
            \deg(V_m), & 1\le i\le m-1,\\[0.2em]
            -(m-1)\deg(V_m), & i=m,\\[0.2em]
            0, & i>m,
        \end{cases}
    \]
    that is,
    \[
        L_n u_m = \deg(V_m)\,u_m.
    \]
    Thus, $u_m$ is an eigenvector with eigenvalue $\deg(V_m)$.  
    The vectors $u_1,\ldots,u_n$ are linearly independent, by construction, and therefore they form a basis of $\mathbb{R}^n$ consisting of eigenvectors of $L_n$, yielding the claim.
\end{proof}

\begin{remark}
    An interesting and distinct feature of the Laplacian spectrum of the P\'olya threshold graph established in Theorem~\ref{theorem:eigenvalue} is that its eigenvalues are random variables, whereas its corresponding eigenvector basis is entirely deterministic. This is because the graph is constructed recursively by adding either isolated or universal nodes, with node types determined by the P\'olya urn draws. As a result, the randomness affects only the eigenvalues, while the deterministic structure of threshold graphs fixes the Laplacian eigenvector basis, in contrast to classical random graph models where both eigenvectors and eigenvalues are random.

    Furthermore, since Theorem~\ref{theorem:eigenvalue} expresses the Laplacian eigenvalues in terms of the node degrees, the probability distribution of the individual Laplacian eigenvalues is fully characterized by the node degree distribution derived in Theorem~\ref{theorem:degreeprob}. This high tractability provides a solid foundation for studying various spectral notions of the network, such as its algebraic connectivity (the second smallest Laplacian eigenvalue, given by $\min_{2\leq i\leq n}\deg(V_i)$), its largest eigenvalue, and empirical spectral distribution.
\end{remark}

\section{Application to consensus}\label{section:consensus}
As an application, following \autocite{degroot1974reaching, olfatisaber2007consensus, shi2015consensus_random_graph_processes}, we analyze the consensus limit for the linear averaging dynamics induced by the P\'olya threshold graph. Consensus dynamics of random graphs is a well studied subject~\autocite{touri2013productrandomstochasticmatrices}, and our objective here is merely to relate the consensus limit to the stochastic properties of the P\'olya threshold graph.
Throughout this section, we consider a \textit{connected} size-$n$ P\'olya threshold graph obtained by fixing $Z_n$ to $1$ (equivalently, by enforcing the last added node to be universal). Hence, $\G_n$ is connected for every $n$. If this condition is not guaranteed, the graph may decompose into a connected component and multiple isolated nodes, and the results apply to the connected component.
\subsection{Consensus limit}\label{section:consensuslimit}
For each $i\in\{1,\ldots,n\}$, we assign an initial opinion $x_i(0)\in\mathbb{R}$ to the node $V_i\in\V_n$. At each discrete time $t\ge 1$, the node $V_i$ updates its opinion $X_i(t)$ by averaging its own previous opinion with those of its neighbors:
\begin{align} 
    X_i(t) &= \frac{1}{N_i}\bigg[X_i(t-1) + \sum_{\substack{j=1, j\neq i}}^n a_{ij}\,X_j(t-1)\bigg]\\
    &= \frac{1}{N_i}\bigg[X_i(t-1) + \sum_{\substack{j=1, j\neq i}}^n Z_{\max(i,j)}\,X_j(t-1)\bigg],\label{equation:consensus}
\end{align}
where $X_i(0) = x_i(0)$ and
\[
    N_i := 1+\sum_{\substack{j=1, j\neq i}}^n Z_{\max(i,j)}
\]
denotes the number of neighbors of $V_i$ in $\G_n$ plus one.

To rewrite the updating rule \eqref{equation:consensus} in a compact matrix form, for all \(i,j\in\{1,\ldots,n\}\), we define the weight
\[
    W_{ij} := \frac{\mathbb{I}(i=j)+\mathbb{I}(i\neq j)Z_{\max(i,j)}}{N_i},
\]
and let \(W=[W_{ij}]\) be the corresponding \(n\times n\) matrix (deterministic given $Z^n=(z^{n-1},1)$ is fixed). Then \eqref{equation:consensus} can be written as
\begin{align*}
    X_i(t)
    &= \frac{1}{N_i}\bigg[X_i(t-1) + \sum_{\substack{j=1, j\neq i}}^n Z_{\max(i,j)}\,X_j(t-1)\bigg] \\*
    &= \sum_{j=1}^n W_{ij}\,X_j(t-1).
\end{align*}
If we denote \(\bar{X}(t)=(X_1(t),\ldots,X_n(t))^\intercal\), then in vector form, we have
\begin{align}\label{equation:x_W}
    \bar{X}(t) &= W \bar{X}(t-1)
         = W^t \bar{x}(0),
\end{align}
where $\bar{x}(0) = (x_1(0), x_2(0),\ldots, x_n(0))^\intercal$.

\begin{lemma}\label{Wirreducibleandaperiodic}
For any fixed draw realization $Z^n=(z^{n-1},1)$, the corresponding matrix \(W=[W_{ij}]\) is an irreducible, aperiodic and stochastic matrix.
\end{lemma}

\begin{proof}
    We first verify that each row of \(W\) sums to \(1\). For any fixed \(i\in\{1,\ldots,n\}\),
    \begin{align*}
        \sum_{j=1}^n W_{ij}
        &= \sum_{j=1}^n \frac{\mathbb{I}(i=j)+\mathbb{I}(i\neq j)z_{\max(i,j)}}{N_i} \\
        &= \frac{1+\sum_{\substack{j=1, j\neq i}}^n z_{\max(i,j)}}{N_i} \\
        &= \frac{N_i}{N_i} \\
        &=1,
    \end{align*}
    which implies that \(W\) is a stochastic matrix.

    Since \(Z_n=1\), node \(V_n\) is adjacent to all other nodes. Thus, for any \(i,j\in\{1,\ldots,n\}\) there is a path from \(V_i\) to \(V_j\) of length at most \(2\) through \(V_n\). Hence \(W\) is irreducible.

    Moreover, $W_{ii}=\frac{1}{N_i}>0$ for all $i$, reflecting that each node always assigns positive weight to its own previous opinion in \eqref{equation:consensus}. Hence, the Markov chain associated with $W$ has a positive self-transition probability for all states, and since $W$ is irreducible, its period is $1$, i.e., $W$ is aperiodic.
\end{proof}

For $i\in\{1,\ldots,n\}$, define
\[
    \pi^*:=(\piz_1,\ldots,\piz_n),
    \qquad
    \piz_i:=\frac{N_i}{\sum_{k=1}^n N_k}.
\]
Here $\pi^*$ is viewed as a row vector, so $\mathbf{1}\pi^*$ is an $n\times n$ rank-one stochastic matrix, where $\mathbf{1}$ is a size-$n$ column vector of $1$s. Following standard results in Markov chain theory \autocite{seneta2006nonnegative}, we have the following property.

\begin{lemma}\label{WhasPi}
    For any fixed draw realization $Z^n=(z^{n-1},1)$, the vector \(\pi^*\) is the unique stationary distribution of \(W\)
    \[
        W^t \xrightarrow[t\to\infty]{} \mathbf{1}\pi^*.
    \]
\end{lemma}

\begin{proof}
    For any \(i,j\in\{1,\ldots,n\}\), by the definition of \(W_{ij}\), we have that
    \begin{align*}
        N_i W_{ij} &= \mathbb{I}(i=j)+\mathbb{I}(i\neq j)z_{\max(i,j)},\\
        N_j W_{ji} &= \mathbb{I}(j=i)+\mathbb{I}(j\neq i)z_{\max(i,j)}.
    \end{align*}
    Since the conditions \(i\neq j\) and \(j\neq i\) are equivalent, it follows that
    \[
        N_i W_{ij} = N_j W_{ji},
    \]
    with
    \[
        \piz_i W_{ij} = \piz_j W_{ji}.
    \]
    Fixing $j$ and summing over \(i\) on both sides yields
    \begin{align}
        \sum_{i=1}^n \piz_i W_{ij}
        &= \sum_{i=1}^n \piz_j W_{ji} \nonumber\\*
        &= \piz_j \sum_{i=1}^n W_{ji} \nonumber\\*
        &= \piz_j,
        \label{equation:MarkovChainStationaryDistri}
    \end{align}
    where we used the fact that each row of \(W\) sums to \(1\). Thus, \(\pi^* W=\pi^*\), i.e., \(\pi^*\) is a stationary distribution of \(W\). By Lemma~\ref{Wirreducibleandaperiodic}, \(W\) is irreducible and aperiodic, hence the stationary distribution is unique and the associated Markov chain is ergodic. Consequently, we have by~\autocite{Serfozo2009Basics} that
 $
        W^t \xrightarrow[t\to\infty]{} \mathbf{1}\,\pi^*,
  $
    as claimed.
\end{proof}

We now use Lemma~\ref{WhasPi} to describe the consensus limit of the process. 

\begin{theorem}\label{theorem:consensuslimit}
    Let \(\G_n=(\V_n,\E_n)\) be a connected P\'olya threshold graph with draw indicator n-tuples \(Z^n=(Z_1,Z_2,\ldots,Z_{n-1},1)\). Then
    \begin{align*}
        \limt E\big(\bar{X}(t)\big) = [\mathbf{1}\pi^{(E)}]\bar{x}(0),
    \end{align*}
where 
$\pie = (\pi^{(E)}_1,\pi^{(E)}_2,\ldots,\pi^{(E)}_n)$
and
    \begin{align}\label{equation:pie}
       \pi^{(E)}_i
        := \sum_{z_1}\cdots\sum_{z_{n-1}}
            P(Z_1=z_1,\ldots,Z_{n-1}=z_{n-1})\,\piz_i.
    \end{align}
\end{theorem}

\begin{proof}
    From \eqref{equation:x_W}, we have \(\bar{X}(t)=W^t \bar{x}(0)\). Since \(\bar{x}(0)\) is deterministic,
    \begin{align*}
        \limt E\big(\bar{X}(t)\big)
        &= \Big(\limt E\big(W^t\big)\Big) \bar{x}(0).
    \end{align*}
    Note that
    \begin{align*}
        \limt E\big(W^t\big)
        &= \limt \sum_{z_1}\cdots\sum_{z_{n-1}}
            P(Z_1=z_1,\ldots,Z_{n-1}=z_{n-1})\,W^t.
    \end{align*}
    Since the sum is over a finite state space $z^{n-1}\in\{0,1\}^{n-1}$, we may interchange the limit and the sum. For each fixed draw vector \(Z^n = (z^{n-1}, 1)\), Lemma~\ref{WhasPi} implies that
    \[
        W^t \xrightarrow[t\to\infty]{} \mathbf{1}\,\pi^*.
    \]
    Therefore,
    \begin{align*}
        \limt E\big(W^t\big)
        &= \sum_{z_1}\cdots\sum_{z_{n-1}}
            P(Z_1=z_1,\ldots,Z_{n-1}=z_{n-1})\,\mathbf{1}\,\pi^*,
    \end{align*}
    and the stated expression for \(\lim_{t\to\infty}E(\bar{X}(t))\) follows immediately.
\end{proof}

Having established the consensus limit, our next goal is to show that $\pie$ is the stationary distribution of the limiting expected averaging matrix, i.e.,
\[
   \pie\Big(\limt E(W^t)\Big)=\pie.
\]
To do so, we first verify that \(\lim_{t\to\infty}E(W^t)\) is itself a stochastic matrix.

\begin{lemma}\label{E(W^t) primitive}
    The limit
    \[
        \limt E(W^t)
    \]
    is a stochastic matrix. Moreover, \(\lim_{t\to\infty}E(W^t)\) is irreducible and aperiodic. 
\end{lemma}

\begin{proof}
    By Theorem~\ref{theorem:consensuslimit} and Lemma~\ref{WhasPi}, we have
    \begin{align}
        \limt E(W^t)
        &= \sum_{z_1}\cdots\sum_{z_{n-1}}
            P(Z_1=z_1,\ldots,Z_{n-1}=z_{n-1})\,\mathbf{1}\,\pi^*.
        \label{equation:E(W^t)}
    \end{align}
    Since $\mathbf{1}\piz$ is a stochastic matrix, the limit $\limt E(W^t)$ is a convex combination of stochastic matrices and hence is itself a stochastic matrix.

    Furthermore, the P\'olya urn process assigns a strictly positive probability to each draw vector \(z^{n-1}\), and \(\piz_i>0\) for every \(i\in\{1\ldots,n\}\) because \(V_n\) is a neighbor of all other nodes. It follows from \eqref{equation:E(W^t)} that every entry of \(\lim_{t\to\infty}E(W^t)\) is strictly positive. Hence $\lim_{t\to\infty}E(W^t)$ is a positive (thus primitive) stochastic matrix, and in particular irreducible and aperiodic.
\end{proof}

Next, by Lemma~\ref{E(W^t) primitive} we obtain the following result.

\begin{theorem}\label{theorem:pi and W^t}
    The vector \(\pie=(\pie_1,\ldots,\pie_n)\) defined by \eqref{equation:pie} is a stationary distribution of the limiting expected averaging matrix \(\limt E(W^t)\), which satisfies
    \[
       \pie\Big[\limt E(W^t)\Big] =\pie.
    \]
\end{theorem}

\begin{proof}
    From \eqref{equation:E(W^t)} we can write the limiting matrix $\limt E(W^t)$ as  $\mathbf{1}\,\pie$. Since $\pie$ is a convex combination of probability vectors $\piz$, it is strictly a probability vector satisfying $\pie\mathbf{1}= 1$. Therefore,
    \begin{align*}
       \pie\Big[\limt E(W^t)\Big]
        &=\pie\big(\mathbf{1}\,\pie\big) = \big(\pie\mathbf{1}\big)\,\pie =\pie,
    \end{align*}
    which completes the proof.
\end{proof}

\begin{remark}\label{remark:consensuslimit}
    We emphasize that the consensus limits established in Theorems~\ref{theorem:consensuslimit} and~\ref{theorem:pi and W^t} are structurally decoupled from the specific joint distribution introduced by the P\'olya urn process. The derivation relies fundamentally on the algebraic structure of the threshold graph's adjacency matrix. Consequently, these results apply to any random threshold graph generated by an arbitrary binary stochastic process $\{Z_i\}^n_{i=1}$, provided that the last node is universal (i.e., $Z_n = 1$).
\end{remark}

\subsection{Simulation results}
In this section, we illustrate the convergence of consensus dynamics on the P\'olya threshold graphs and validate the consensus limit obtained in Theorem~\ref{theorem:consensuslimit}. We first consider the standard (i.e., infinite-memory) model we examined throughout this paper, followed by an exploration of a finite-memory variant to highlight the impact of the memory length on the consensus limit. 

\subsubsection{P\'olya threshold graph}\label{Infinite-memory process}
We first fix the graph size to $n = 10$ and the P\'olya urn parameters to $(R,B,\Delta)=(5, 5, 2)$. In each run, we generate the graph by simulating the P\'olya urn process for the first $n-1$ nodes, and deterministically assign the $n$-th node as a universal node. We construct the corresponding P\'olya threshold graph as in Lemma~\ref{lemma:thresholdgraph}, and initialize the opinions as $(0.1, 0.6, 0.3, 1, 0.5, 3, 10, 2, 9, 0.2)^\intercal$. The consensus then evolves according to \eqref{equation:consensus}. Over 200 independent runs, we compute the empirical average of the consensus limit and compare it with the theoretical prediction from Theorem~\ref{theorem:consensuslimit}.

Figure~\ref{fig:traj_n10} shows one realization of a P\'olya threshold graph with $n = 10$ and the draw vector
$z^{10} = (0,0,1,0,1,1,0,1,0,1)$.
Trajectories of the isolated nodes are plotted with dashed lines, while those of universal nodes are plotted with solid lines.
We can see that nodes with larger degrees converge more rapidly toward the consensus value, which is consistent with the intuition
that highly connected nodes place and receive more weight in the averaging dynamics.
Figure~\ref{fig:hist_n=10} shows the distribution of consensus values at time $t = 100$ over $200$ independent simulations with the above setup.
We use a dashed vertical line to represent the empirical average of the $200$ consensus values, while the solid vertical line indicates the theoretical value $\pie \bar{x}(0)$.
The empirical mean closely matches the theoretical value $\pie \bar{x}(0)$, which is consistent with Theorem~\ref{theorem:consensuslimit}.

Finally, we repeat the same experiment with $n=100$ (still fixing $z_{100}$ to be $1$) to compare with the previous simulation. We use the same initial opinions as in the $n = 10$ case and extend them to length $100$ with repetition: $x_{i+10k}(0)=x_i(0)$, $i\in\{1,\ldots,10\}$, $k\in\{0,\ldots,9\}$. Equivalently, $x_i(0)=x_{i+10}(0)=x_{i+20}(0)=\cdots=x_{i+90}(0)$ for all $i\in\{1,\ldots, 10\}$. Figure~\ref{fig:traj_n100} shows the consensus trajectories of different nodes, mirroring the curves shown for the size-10 graph in Figure~\ref{fig:traj_n10}.
Figure~\ref{fig:consensus_limits_vs_pi(n=100)} gives the corresponding histogram of consensus limits over $200$ runs, together with the same two vertical references sample mean and $\pie \bar{x}(0)$.
This provides a direct comparison between the cases $n=10$ and $n=100$ with identical urn parameters and a consistent choice of $\bar{x}(0)$.

\begin{figure}[htbp]
    \centering
    \begin{subfigure}[t]{0.59\textwidth}
        \centering
        \includegraphics[width=\linewidth]{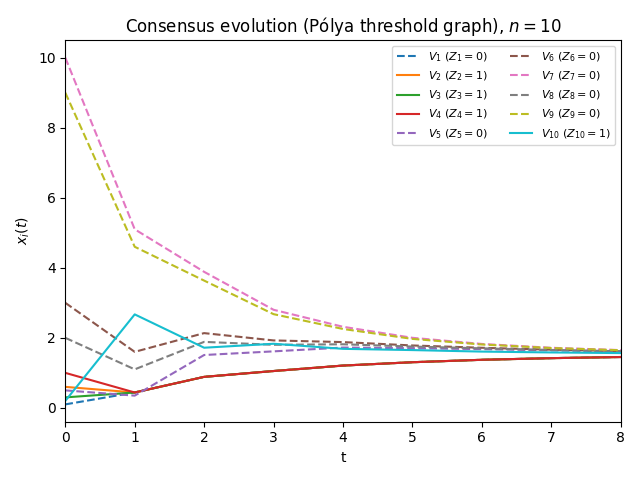}
    \end{subfigure}\hfill%
    \begin{subfigure}[t]{0.34\textwidth}
        \centering
        \resizebox{\linewidth}{!}{%
        \begin{tikzpicture}[
            every node/.style={circle, draw, minimum size=6mm, inner sep=0pt, font=\small},
            universal/.style={fill=black!20, very thick},
            isolated/.style={fill=white},
        ]
        \foreach \i/\z in {1/0,2/0,3/1,4/0,5/1,6/1,7/0,8/1,9/0,10/1}{
            \pgfmathsetmacro{\ang}{90 - 36*(\i-1)}
            \ifnum\z=1
                \node[universal] (v\i) at (\ang:2.2cm) {\i};
            \else
                \node[isolated] (v\i) at (\ang:2.2cm) {\i};
            \fi
        }
        \begin{scope}[on background layer]
            \foreach \j in {3,5,6,8,10}{
                \pgfmathtruncatemacro{\k}{\j-1}
                \foreach \i in {1,...,\k}{
                    \draw (v\i) -- (v\j);
                }
            }
        \end{scope}
        \end{tikzpicture}%
        }
    \end{subfigure}

    \caption{Consensus trajectories for one realization of a P\'olya threshold graph with $n=10$, urn parameters $(R,B,\Delta)=(5,5,2)$, and draw vector $z^{10}=(0,0,1,0,1,1,0,1,0,1)$. The initial opinions are fixed as $\bar{x}(0)=(0.1,0.6,0.3,1,0.5,3,10,2,9,0.2)^\intercal$. Dashed curves and white nodes correspond to isolated nodes; solid curves and dark nodes correspond to universal nodes.}
    \label{fig:traj_n10}
\end{figure}
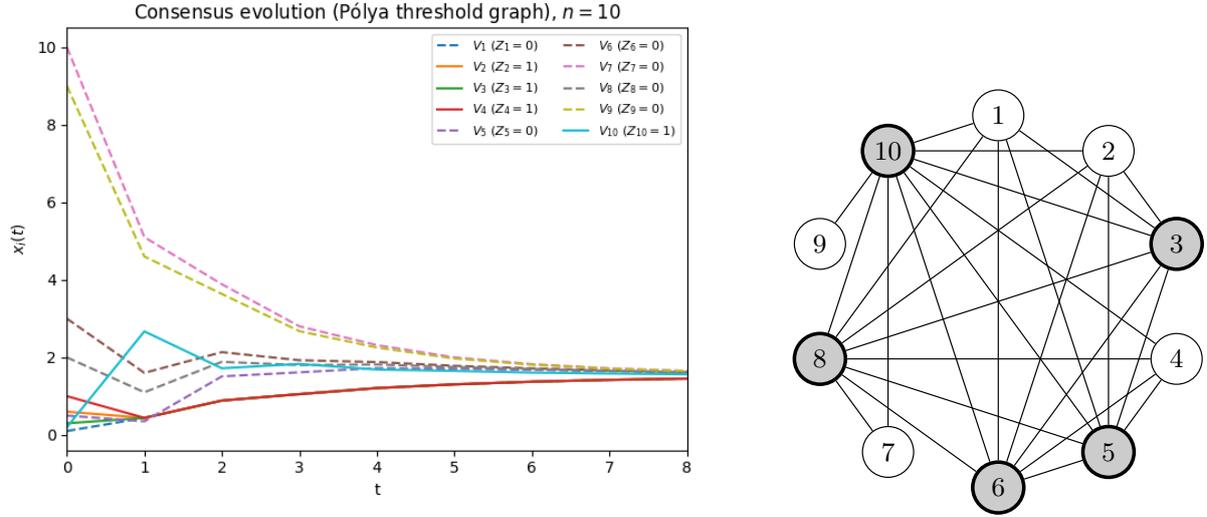

\begin{figure}[htbp]
    \centering
    \includegraphics[width=0.7\linewidth, trim=0 10 0 10, clip]{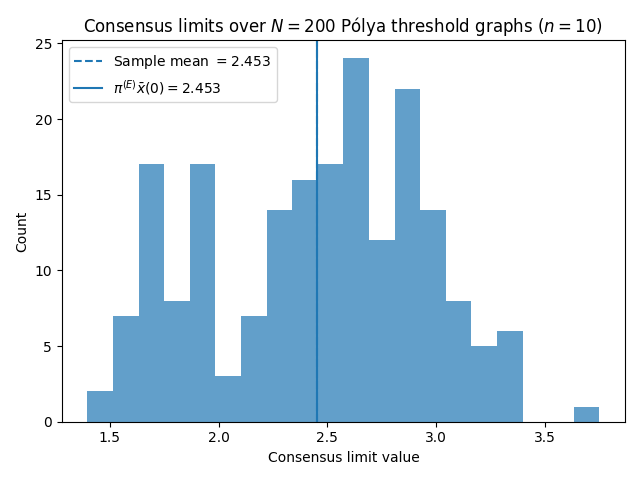}
    \caption{Histogram of consensus values at time $t=100$ over $200$ independent simulations with $n=10$ and urn parameters $(R,B,\Delta)=(5,5,2)$, with $z_{10}$ fixed to be $1$. The initial opinion vector is fixed as $\bar{x}(0)=(0.1,0.6,0.3,1,0.5,3,10,2,9,0.2)^\intercal$. The dashed vertical line indicates the sample mean of the consensus values, and the solid vertical line indicates the theoretical prediction $\pie \bar{x}(0)$.}
    \label{fig:hist_n=10}
\end{figure}

\begin{figure}[htbp]
    \centering
    \begin{subfigure}[t]{0.49\textwidth}
        \centering
        \includegraphics[width=\linewidth]{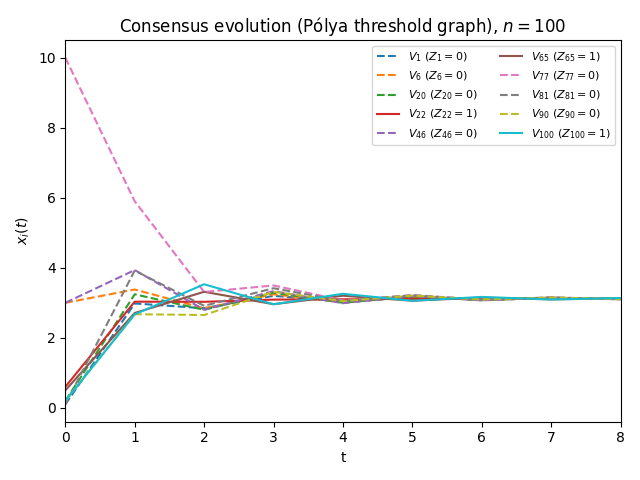}
        \caption{Consensus trajectories for one realization of a P\'olya threshold graph with $n=100$, urn parameters $(R,B,\Delta)=(5,5,2)$, and draw vector $z^{100}$.}
        \label{fig:traj_n100}
    \end{subfigure}\hfill%
    \begin{subfigure}[t]{0.49\textwidth}
        \centering
        \includegraphics[width=\linewidth]{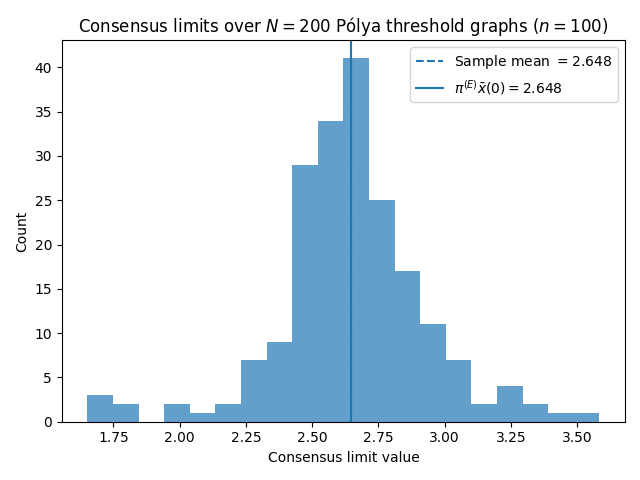}
        \caption{Histogram of consensus limits with $n=100$ and $(R,B,\Delta)=(5,5,2)$, with $z_{100}$ fixed to be $1$. The dashed vertical line is the sample mean of the simulated consensus limits, and the solid vertical line is the $\pie \bar{x}(0)$.}
        \label{fig:consensus_limits_vs_pi(n=100)}
    \end{subfigure}
    \caption{Simulation results for a P\'olya threshold graph with $n =100$, urn parameters $(R,B,\Delta)=(5,5,2)$, and draw vector $z^{100}$ (fixed $z_{100}=1$). }
    \label{fig:n100_combined}
\end{figure}

\FloatBarrier

\subsubsection{Finite-memory P\'olya threshold graph}
We also consider a finite memory version of the P\'olya urn introduced in~\autocite{alajaji2002communication}, a variation often studied in the literature of networked P\'olya urn processes, see~\autocite{singh-consensus22,singh-siam22}. In the P\'olya threshold graph model discussed previously, every draw permanently alters the urn's composition and, consequently, the probability of all future draws, thereby constituting an infinite-memory process. In the finite-memory variant, the only difference is that the $\Delta$ reinforcement balls added at time step $t$ are removed from the urn exactly $M$ steps later (i.e., at time step $t+M$). Consequently, when $t\leq M$, the probability distribution remains unchanged as that of the (infinite-memory) P\'olya urn process given in~\eqref{jointprobofPolya_product}-\eqref{jointProbofPolya}, as no reinforcement balls have been removed yet (thus the graph generated for $t\le M$ is the same as the original P\'olya graph).

However, when $t>M$, this finite-memory setup fundamentally alters the underlying stochastic dependence of the draw process $\{Z_t\}$. Under this finite-memory assumption, the draw process $\{Z_t\}$ becomes a Markov process of order $M$, as demonstrated in~\autocite{alajaji2002communication}. More specifically, the probability distribution is explicitly given by
\begin{align}
P(Z_1=z_1, \ldots, &Z_n=z_n) = \left( \frac{\prod_{i=0}^{\left(\sum_{m=1}^M z_m\right)-1}(\rho+i\delta) \prod_{j=0}^{M-1-\sum_{m=1}^M z_m}(1-\rho+j\delta)}{\prod_{l=1}^{M-1}(1+l\delta)} \right)\\
&\times\prod_{i=M+1}^{n} \left[ \frac{\rho + \delta \sum_{m=1}^M z_{i-m}}{1+M\delta} \right]^{z_i} \left[ \frac{1-\rho + \delta \left(M - \sum_{m=1}^M z_{i-m}\right)}{1+M\delta} \right]^{1-z_i}, \nonumber
\end{align}
for any $(z_1,\ldots,z_n)\in \{0,1\}^n$.

\medskip

As stated in Remark~\ref{remark:consensuslimit}, since the network retains the threshold graph structure, the analytical framework established in Theorem~\ref{theorem:consensuslimit} remains fully applicable.  
To evaluate the consensus dynamics under this finite-memory regime, we simulate the connected P\'olya threshold graph of size $n=10$ and $n = 100$ with memory step $M$. To satisfy the connectivity condition $Z_n =1$, we generate the graphs by simulating first $n-1$ nodes under the finite-memory process with the probability described above, and deterministically setting the final node to be a universal node (i.e., $Z_n = 1$). To capture the impact of memory length, we choose the initial opinion $\bar{x}(0)$ to be a polarized vector, assigning $0$ to the first half of $\bar{x}(0)$ and $100$ to the second half. By assigning maximally different initial opinions, any structural shift caused by memory length $M$ becomes visible in the final consensus. For each $M$ value, we compute the empirical expected consensus $\pie\bar{x}(0)$ by averaging 100 independent random threshold graphs.

Figures~\ref{fig:piE_vs_M_n10} and~\ref{fig:piE_vs_M_n100} plot the expected consensus versus the memory length $M$ for $\delta\in\{0.2, 1, 10\}$, along with their corresponding infinite-memory baseline (dashed lines). As observed in both figures, the impact of the memory length $M$ is sensitive to the reinforcement parameter $\delta$. For weak reinforcement (e.g., $\delta = 0.2$), the finite memory curves remain close to the infinite-memory baseline across all values of $M$, as the delayed removal only slightly affects the urn's overall composition. Conversely, under strong reinforcement (e.g., $\delta = 10$), the finite-memory curve deviates significantly as $M$ increases. In all cases, as $M$ increases and approaches the size of the graph $n$, the consensus limit naturally converges to the infinite-memory values.

\begin{figure}[htbp]
    \centering
    \begin{subfigure}[t]{0.48\textwidth}
        \centering
        \includegraphics[width=\linewidth]{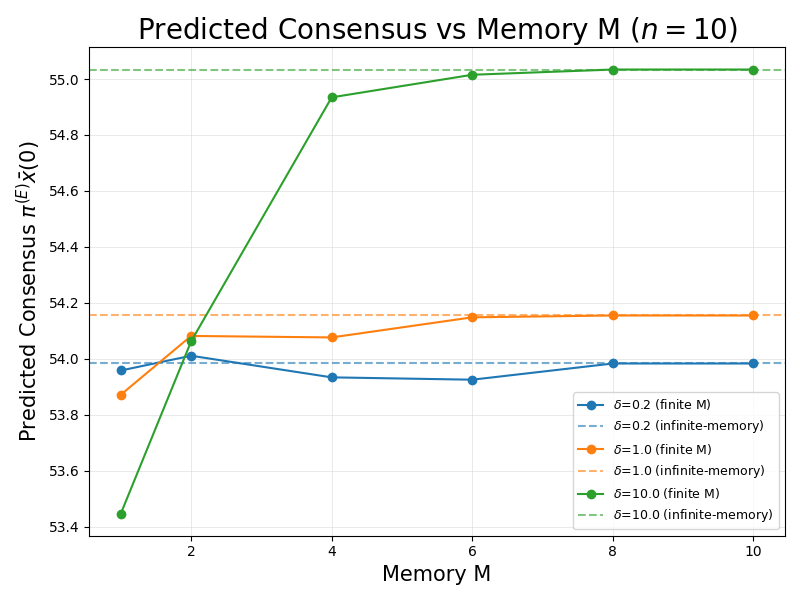}
        \subcaption{$n=10$}
        \label{fig:piE_vs_M_n10}
    \end{subfigure}\hfill%
    \begin{subfigure}[t]{0.48\textwidth}
        \centering
        \includegraphics[width=\linewidth]{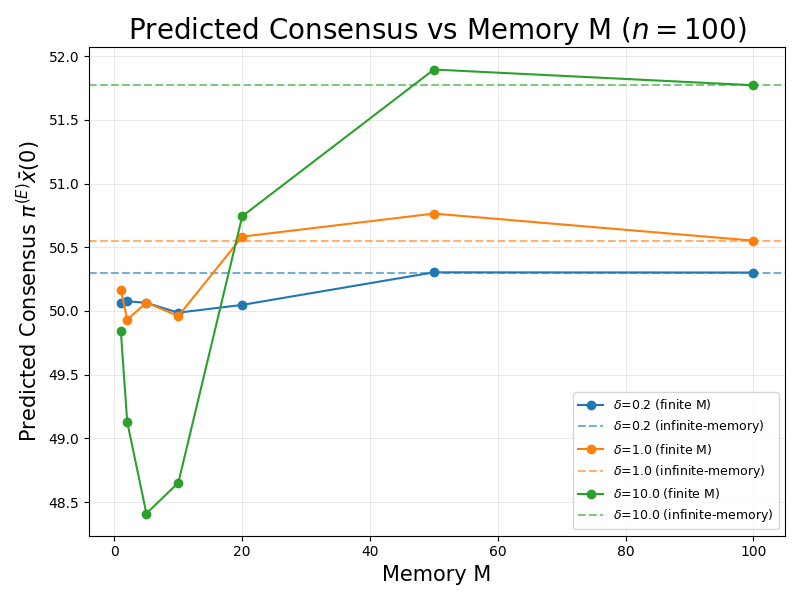}
        \subcaption{$n=100$}
        \label{fig:piE_vs_M_n100}
    \end{subfigure}

    \caption{$\pi^{(E)}\bar{x}(0)$ versus memory length $M$ for $n=10$ and $n=100$ (solid: finite-memory; dashed: infinite-memory).}
    \label{fig:piE_vs_M_both}
\end{figure}
\FloatBarrier

\section{Conclusion}\label{sec:conclusion}
In this paper, we introduced the P\'olya threshold graph model, a random graph construction that combines the nested structure of threshold graphs with the reinforcement mechanism of the two-color P\'olya urn process. We analytically derived its stochastic and algebraic properties. More specifically, based on the exchangeability and Beta-Binomial structure of the P\'olya urn process, we obtained the degree distribution of any node in a given P\'olya threshold graph, as well as its mean and variance. Furthermore, we derived an explicit formula for the expected decay centrality score.
Finally, we established the Laplacian matrix of the P\'olya threshold graph and characterized its spectrum and corresponding eigenbasis.

As an application of these results, we analyzed discrete-time linear consensus dynamics over connected P\'olya threshold graphs. We demonstrated that the consensus dynamics yield an irreducible and aperiodic averaging matrix, allowing us to fully characterize the limiting expected consensus value. The simulation results strongly align with our theoretical predictions, illustrating how highly connected nodes converge more rapidly toward the consensus value. Additionally, by exploring a finite-memory P\'olya urn variant, we illustrated how the length of memory affects the resulting expected consensus.

Several directions for future research naturally arise from this work. One possible extension is the study of P\'olya threshold graphs generated by multi-color P\'olya urn models, which would allow for more general node-connection mechanisms. Another promising direction is the analysis of stochastic properties of the finite-memory model, including the investigation of its Markovian structure and the dependence of convergence rates on the memory size.


\bibliographystyle{plainnat}
\bibliography{bib}

\end{document}